%% file: paper.tex
\documentclass[letterpaper,twocolumn,10pt]{article}

\PassOptionsToPackage{table,dvipsnames}{xcolor}

\usepackage{usenix-2020-09}

\usepackage[english]{babel}
\usepackage{amsmath}
\usepackage{amssymb}
\usepackage{amsfonts}
\usepackage{amsthm}
\usepackage{graphicx}
\usepackage[capitalize,noabbrev]{cleveref}
\usepackage{xspace}
\usepackage{ulem}
\usepackage{booktabs}
\usepackage{tabularx}
\usepackage{algorithm}
\usepackage{algpseudocode}
\usepackage{bm}
\usepackage{caption}
\usepackage{enumitem}
\usepackage{subcaption}
\usepackage{comment}   

\usepackage[compact]{titlesec}
\titlespacing*{\section}       {0pt}{1.2ex plus .3ex minus .2ex}{0.6ex plus .1ex}
\titlespacing*{\subsection}    {0pt}{1.0ex plus .2ex minus .2ex}{0.4ex plus .1ex}
\titlespacing*{\subsubsection} {0pt}{0.8ex plus .2ex minus .2ex}{0.3ex plus .1ex}

\setlist{topsep=2pt,itemsep=1pt,parsep=1pt}

\renewcommand{\abstractname}{\normalsize Abstract}

\let\citet\cite

\newtheorem{theorem}{Theorem}

\theoremstyle{definition}

\begin{document}
\title{GATE: GPU-Accelerated Traffic Engineering for the WAN}

%
\author{
{\rm Rahul Bothra$^{1, 2}$ \quad Alexander Krentsel$^{1, 3}$ \quad Saptarshi Mandal$^2$ \quad Brighten Godfrey$^2$} \\
{\rm Sylvia Ratnasamy$^{1, 3}$ \quad Rob Shakir$^1$ \quad R. Srikant$^2$}\\
\\
$^1$Google \quad $^2$UIUC \quad $^3$UC Berkeley
}
\date{}

\newcounter{rescount}

\newcommand{\result}[2][]{%
  \refstepcounter{rescount}%
  \noindent\textit{Result \#\therescount: #2\\}
  \if\relax\detokenize{#1}\relax\else\label{#1}\fi%
}


\renewcommand{\algorithmiccomment}[1]{\hfill\footnotesize\textcolor{gray}{#1}}


\newcommand{\cut}{\sout}
\newcommand{\topic}[1]{\textcolor{black}{\bf #1}}
\newcommand{\red}[1]{\textcolor{red}{#1}}
\newcommand{\fullsysname}{GPU-Accelerated Traffic Engineering\xspace}
\newcommand{\sys}{GATE\xspace}
\newcommand{\wansmall}{WAN A\xspace}
\newcommand{\wanlarge}{WAN B\xspace}
\newcommand{\sysname}{\sys}
\newcommand{\obj}[1]{\textbf{(O.#1)}\xspace}
\newcommand{\bg}[1]{\textcolor{blue}{[Brighten: #1]}}
\newcommand{\bgc}[1]{\textcolor{blue}{[Brighten: #1]}}
\newcommand{\bgnew}[1]{\textcolor{blue}{#1}}
\newcommand{\ak}[1]{\textcolor{brown}{[Alex: #1]}}
\newcommand{\aknew}[1]{\textcolor{black}{#1}}
\newcommand{\srnew}[1]{\textcolor{violet}{[Sylvia: #1]}}
\newcommand{\rb}[1]{\textcolor{orange}{[Rahul: #1]}}
\newcommand{\rbnew}[1]{\textcolor{orange}{#1}}
\newcommand{\geant}{G\'EANT\xspace}
\newcommand{\osr}[1]{\textcolor{purple}{#1}}
\newcommand{\gem}[1]{\textcolor{red}{#1}}
\newcommand{\osrdel}[1]{\textcolor{purple}{\sout{#1}}}
\newcommand{\gemdel}[1]{\textcolor{red}{\sout{#1}}}

\renewcommand{\osrdel}[1]{}

\renewcommand{\bg}[1]{\textcolor{blue}{}}
\renewcommand{\rbnew}[1]{\textcolor{black}{#1}}
\renewcommand{\osr}[1]{\textcolor{black}{#1}}
\renewcommand{\bgnew}[1]{\textcolor{black}{#1}}

\renewcommand{\bgc}[1]{}
\renewcommand{\ak}[1]{}
\renewcommand{\rb}[1]{}
\renewcommand{\srnew}[1]{}

\crefformat{section}{\S#2#1#3}
\crefformat{subsection}{\S#2#1#3}
\crefname{equation}{Eq.}{Eqs.}
\Crefname{equation}{Eq.}{Eqs.}
\crefname{theorem}{Thm.}{Thms.}
\Crefname{theorem}{Thm.}{Thms.}
\crefname{figure}{Fig.}{Figs.}
\Crefname{figure}{Fig.}{Figs.}

\Crefformat{section}{\S\S#2#1#3} 
\Crefformat{subsection}{\S\S#2#1#3} 

\renewcommand{\topic}[1]{#1}

\maketitle

\input{0-abstract}

\input{1-introduction}
\input{2-motivation}
\input{3-design}
\input{4-evaluation}
\input{6-conclusion}

\bibliographystyle{abbrv}
\bibliography{reference}
\input{7-appendix}
\end{document}

%% file: 0-abstract.tex
\begin{abstract}
Traffic engineering (TE) has become a crucial tool for enforcing routing policy and maintaining operational efficiency \aknew{in large networks.}
Existing TE solutions pick an objective function to optimize, aiming to balance (i) allocating traffic optimally with (ii) reacting quickly to demand changes and disruption events. However, as the scale of networks grows, the runtime of the existing optimal solution becomes infeasibly large. The alternative -- approximate solvers -- result in costly inefficiencies, \osr{O(\$100M), on large production WANs}.

We present \fullsysname (\sys), which achieves the best of both worlds: enabling fast TE runtimes through a highly-parallelizable GPU-compatible decomposition, while iteratively converging to the provably optimal solution. \sys unlocks a unique set of desirable properties: it becomes increasingly parallelizable with network size, supports a wide spectrum of fairness objectives, and offers theoretically guaranteed convergence to the optimal solution and near-optimal convergence within a bounded time.
We evaluate \sys on production traces from two large cloud WANs, and show that \sys achieves near-optimal solutions \osrdel{5}\osr{4}-10x faster than state-of-the-art.
\end{abstract}

%% file: 1-introduction.tex
\section{Introduction}
\label{sec:introduction}

Traffic engineering (TE) is central to maximizing utilization and meeting performance goals in wide-area network (WAN) deployments. Large-scale networks are expensive to build and maintain, requiring a large physical footprint and expensive hardware. While traditional shortest-path routing \osr{(realized by OSPF, BGP, etc.)} \osrdel{such as that which is enabled by IS-IS and BGP }can route traffic according to a predefined set of metrics, they do not consider other criteria (e.g., capacity) which can leave parts of the network highly underutilized, translating to wasted costs for network providers. Thus, large-scale cloud providers and network operators deploy TE systems~\cite{ebb, onewan, b4, swan} to have more control over their traffic placement.


\topic{The goal of a traffic engineering system is to (1) place traffic optimally with respect to the operator's goals, and (2) converge to a solution quickly.} Optimality is important because it directly captures either network operations costs (utilization) or customer experience (e.g. throughput, latency, availability). Fast runtime is important because the TE system must react quickly to changing network conditions, such as link cuts or traffic spikes, to bring the network back into an optimal state. While the optimization objective can vary, in practice for many operators~\cite{b4, ebb, danna2017upwardmaxminfair} the objective is \textit{max-min fairness}~\cite{bertsekas2009maxminfair, nace2009maxminfairoverview}, which balances maximizing utilization with maintaining fairness in resource allocation across users and flows.

\topic{Solutions to date have involved a tradeoff between these two goals of runtime and optimality.} This is because direct \aknew{optimal} formulations of the max-min fairness optimization goal (using either general-purpose LP solvers or special-purpose algorithms~\cite{Danna}) are computationally expensive, leading to runtimes of up to 5 minutes on the WANs we study in our evaluation. As a result, most approaches involve formulating the traffic engineering optimization problem as an approximation, sacrificing some optimality in exchange for faster runtimes. Approximate TE solvers have been deployed in production from the earliest SDN WAN deployments, such as the Waterfill algorithm reported in Google's B4 paper~\cite{b4} or the LP approximation in Microsoft's SWAN~\cite{swan}. Recent works on ML-based approximate
solvers~\cite{teal, perry2023dote} show promise in decreasing runtimes even further, but don't support the practical objectives we focus on \aknew{(discussed in \cref{sec:motivation:te-objectives})}, so they won't be a focus of this paper. 
\bgnew{Another line of work has attempted to parallelize TE solving~\cite{ncflow,pop,teal,DeDe}, but also generally compromises optimality guarantees or does not support max-min fairness due to the difficulty of splitting TE into cleanly separable subproblems (discussed in more detail in \S\ref{sec:motivation-parallelization}).}
The degree of acceptable runtime vs. optimality is determined by operator preference based on their workload needs, customer sensitivity, and budgets; thus, each operator tunes their deployment and method to their needs.


As networks are being rapidly scaled to meet AI demands, navigating this tradeoff is becoming increasingly problematic. First, larger networks lead to increasingly longer algorithm runtime, which stress existing formulations that scale superlinearly with input size. Second, larger networks lead to more frequent changes from failures, shifting demand, etc., which requires TE to run more frequently to react.


We propose a new TE approach that enables both fast runtimes and convergence to optimal allocations. Our approach is based on a novel application of \osrdel{Lagrange Decomposition (LD)}\osr{a well-known decomposition} technique\osrdel{s}, \osr{Alternating Direction Method of Multipliers (ADMM)\cite{neal2011distributed},} 
which allows us to break the global TE problem into a collection of smaller per-flow problems that exchange information and make simple updates, iterating until convergence. This process is guaranteed to converge to the optimal allocation, or can be run for a shorter time until it has converged sufficiently closely to optimal. We go beyond a straightforward application of \osrdel{LD}\osr{ADMM} to introduce the notion of \emph{adaptive learning rates} that we show are critical to achieving fast convergence to optimal allocation. \bgc{Too much was deleted in prev sentence -- doesn't parse} Furthermore, we formulate the decomposition as very simple closed-form subproblems. \bgc{Now that I read this again with the modifications, it feels like we are too brief in talking about our innovations. It is phrased as if the previous two sentences are small modifications of an established method, and the properties in the paragraph after this one are simply consequences that we get with no additional work. Maybe we should convey a bit more of the work we did in our new algorithms/designs, either in the previous 2 sentences or in the paragraph that follows.} \srnew{I think it's because the discussion of adaptive learning rates was cut/lost?}\rb{Yes, fixed}

\aknew{These design decisions give \sysname} three important properties that enable efficient implementations. (1) The closed-form simplicity of our decomposition formulation is amenable to a CUDA-kernel implementation over GPUs, enabling massive parallelism. (2) The dynamics of our overall allocation support differential computing, i.e. when processing new information, we can start with the previous allocation to converge much more quickly to the new optimal. 
Finally, (3) our formulation of the decomposition supports practical WAN TE constraints \bgnew{(per-flow demands and pathing)} and a spectrum of objectives, including proportional \bgnew{fairness} and the commonly used objective -- max-min fairness. To our knowledge, \textbf{GATE is the first GPU-accelerated WAN TE solution \bgnew{supporting such practical constraints} with an optimality guarantee.} \bg{Please see what you think of that version. I did insert a few more words but I feel it reads OK -- not too many qualifiers -- and this may better protect us from readers feeling like we are over-claiming.}


\topic{Our evaluation shows that \sysname achieves a better runtime-optimality tradeoff than existing state of the art (SOTA) TE algorithms.} Compared to SOTA TE for optimality (Danna~\cite{Danna}), \sys improves runtime by \textasciitilde\osr{4}\osrdel{1}00x while achieving within $5$\% of the optimal allocation. Compared to \osrdel{the commonly used TE for runtime}\osr{an approximate TE solver used in practice} (k-Waterfill~\cite{k-waterfill}),
\sys achieves a $4$\% improvement in optimality within the same runtime. \aknew{Based on both industry reports and our discussions with cloud providers, at today’s scale, single-digit improvements in efficiency percentages can translate to significant (>O(\$100M)) savings in annual network provisioning costs~\cite{delloro2026capex}.}

\osrdel{To capture the cumulative improvement, we introduce \emph{drift-adjusted optimality (DAO)} as a single metric that captures the joint effect of runtime and optimality, which stems from lower runtimes resulting in more up-to-date routes. }
\osr{While past works evaluate optimality and runtime separately (as in the results quoted above), we observe that runtime is important because lower runtimes result in more up-to-date routes, which in turn improves optimality. We therefore introduce a new metric to capture the combined effect of runtime and optimality, which we call \emph{drift-adjusted optimality (DAO)}.} \sys achieves $4$--$6$\% better DAO than existing TE algorithms (Soroush, SWAN \cite{soroush, swan}) on current topologies. Our results also highlight the importance of our component design contributions to making \osrdel{LD}\osr{decomposition}
techniques practical for TE: adaptive learning enables a 2-3x speedup in runtime compared to a \osrdel{straightforward}\osr{native} application of \osrdel{LD techniques}\osr{the same technique}~\cite{DeDe}, our GPU implementation enables a further 15-20x speedup, and \aknew{\osr{the ability to }warm-start} enables an additional 4-6x speedup.




%% file: 2-motivation.tex
\section{Background and Motivation}
\label{sec:motivation}

Over the past few decades, traditional ISPs and modern hyperscalers alike have adopted SDN-based control architectures to manage their WANs~\cite{swan,ebb,b4,road-to-sdn}. These SDN systems are responsible for collecting state about the network such as topology and demand and programming routers accordingly. At the heart of such systems is an SDN controller that is responsible for deciding how traffic should be placed, given a view of the topology and demand matrix, by running a TE algorithm which produces path allocations for all flows. These inputs to the TE algorithm are aggregated from low-level telemetry such as interface statuses and sending rates which are continually streamed from the routers or hosts in the network. TE is invoked as often as every few minutes, with each invocation taking the latest available information as input. 

\aknew{The end-to-end convergence time of SDN systems is made up of (1) propagation time, during which the controller learns of state changes, (2) computation time, during which the TE algorithm computes new paths, and (3) programming time, during which these paths are programmed into the network. Historically, the computation time and programming time have dominated end-to-end convergence. Recent advances in dataplane techniques such as source routing have been shown to significantly decrease programming time~\cite{dsdn} in a way that stays constant with network size, however computation time remains significant, especially in larger networks.}


\subsection{Traffic Engineering Objectives}
\label{sec:motivation:te-objectives}

\topic{Network operators balance a variety of goals in choosing and configuring their traffic engineering deployments.} These goals are driven by customer needs: providing \obj{1} as much bandwidth as requested, \obj{2} with minimal latency paths, \osrdel{\obj{3} at a low cost, }and \obj{3} with minimal disruptions. Finally, operators aim to do this \osr{\obj{4} at a low cost,} while \obj{5} maintaining fairness across their customers.

\topic{Meeting these goals requires operators to consider both the static and dynamic states of the network.} For a given static snapshot of the network and demands, operators aim to place as much traffic as possible, using shortest paths for higher priority traffic first, then placing lower-priority traffic on longer paths with more bandwidth. Maximizing network utilization minimizes network ownership costs. To avoid unfairly starving any users of bandwidth, this traffic placement must happen with fairness in mind. The most common TE optimization formulation to meet these goals is \textit{max-min fairness}~\cite{nace2009maxminfairoverview}, which aims to give all flows equal bandwidth on their best available paths, while spreading remaining bandwidth to those who need more. 


\topic{In reality, networks are highly dynamic systems, facing constant changes in both traffic demand as application needs change, and network topology as links fail.} Prior work reports average demand swings of $35\%$ per $5$ minutes~\cite{ncflow, teal}. We observed \osrdel{O(100s)}
\osr{hundreds} of link capacity changes each day in the large cloud WAN we analyzed -- which can cause traffic placed on paths traversing affected links to be dropped due to congestion -- in addition to link failures. Network operators deploy a variety of techniques to lessen the impact, such as establishing \osr{MPLS }Fast-Reroute \osrdel{Bypass}
paths to route traffic around downed links; these backup paths place traffic in a manner that avoids dead-ends, but often diverges from the calculated traffic engineering solution’s objectives~\cite{dsdn}, necessitating recalculation to return to meeting operator goals. Even in regular times of operation without faults, demand naturally shifts over time causing traffic placement to drift further and further from optimal.

\topic{For this reason, operators continually re-run TE to keep routing close to optimal.} Yet most existing solvers restart from scratch each invocation, with limited warm-start benefit, so faster runtimes translate directly into fresher allocations. \bgc{This argument needs refinement. Warm-start is a non-sequitur here: regardless of whether you warm-start or how much warm-start helps, faster runtimes translate directly into fresher allocations. Do we need to bring warm start into the picture here? If not, can we remove it?} \srnew{$+1$; either remove or discuss completely. If the latter, first say that faster runtimes enable staying close to optimal. Then say that runtime can be further reduced via optimizations such as warm-start and explain what that means. Then say that GATE is capable of both -  its baseline runtime is fast (competitive/better than deployed TE solutions designed for fast runtime) and it is amenable to a warm-start (not supported by existing solvers). }

\subsection{Runtime-Optimality Tradeoff}
\label{sec:motivation:runtime-optimality-tradeoff}

\topic{Exact formulations of the max-min fairness optimization goal commonly selected by network operators are computationally expensive~\cite{Danna}, with runtimes of $>$5 minutes \osrdel{on standard sized deployments}\osr{for large cloud WANs\cite{swan,b4}}.}
\osrdel{This is far beyond an acceptable time for operators, for whom it is important that the network can react and repair itself after a link or router failure quickly enough to mitigate user impact.}\osr{This duration far exceeds acceptable thresholds for operators, who require rapid repair to mitigate user impact during link or router failures.}
Long runtimes lead to a larger deviation between the TE solver's solution and the new  network state. As a result, practitioners typically use an approximation of the exact solution, sacrificing some optimality in exchange for faster runtime. 

\topic{Accelerating growth in the size and footprints of global-scale WANs has made revisiting the TE problem worthwhile and important.} Large-scale cloud WANs report network size growth of $4$-$7\times$ over the past five years~\cite{azure-network-size-blog, meta-network-size, google-network-size}. This growth is only expected to accelerate in projections of the next five years, driven by the needs of large-scale data processing and machine learning jobs. These larger networks increase infrastructure cost proportionally, making traffic placement efficiency increasingly valuable. At the same time, runtimes using existing TE approaches grow exponentially longer with network size (shown in \cref{s:evaluation:scaling-behavior}), stressing acceptable network reaction times.

\topic{Network growth has made selecting a good operating point balancing runtime and optimality increasingly difficult.} Existing TE algorithms each provide their own tradeoff point. For instance, k-Waterfill~\cite{k-waterfill} solves quickly (\textasciitilde2 sec) but gets a very suboptimal solution ($92$\% of optimal). On the other hand, the exact solver computes the optimal solution, but takes \osrdel{a long time}\osr{much longer} (\textasciitilde$350\text{ sec}$).
The deviation within the traffic matrix in that time is \textasciitilde$20$\%,  making the solution suboptimal. Neither is desirable to the operator.

\topic{To quantify the joint impact of runtime and optimality of a TE solution, we introduce a new metric, Drift Adjusted Optimality (DAO), that combines the \emph{optimality} penalty of the solution with the staleness penalty caused by longer \emph{runtime}}. \srnew{Should we define DAO before going into the following?} The fundamental challenge of existing approximations is that they struggle from one of two challenges that lead to poor DAO: (i) slow reaction to bursts / link-cuts, or (ii) suboptimal solution with low runtimes. As network sizes grow, the DAO of existing solutions deteriorates as runtimes worsen substantially (as existing TE algorithms' runtime grows super-linearly with network size~\cite{pop, k-waterfill}) even if optimality remains the same. Thus, we set our goal to design a traffic engineering algorithm that keeps DAO high as network size grows.




\topic{A secondary goal of our design is \textit{flexibility} in runtime-optimality tradeoff points}. The goal of TE is to reach optimal routing quickly, but ``optimal'' is defined by the operator's business needs and constraints, both of which can change due to shifting business priorities or environmental changes. Whereas existing solutions require configuring a single point in \osr{the} design space, trading off optimality and runtime ahead of deployment, we aim to offer operators flexibility through dynamically configuring a single TE solution at runtime to their preferred tradeoff point. 

\subsection{Related Work}
\label{sec:motivation-existing}

Traffic engineering is a multicommodity flow problem, which (for linear objective functions or max-min fairness) can be formulated as a linear program (LP) optimization problem and solved with general-purpose solvers like Gurobi, or with special-purpose algorithms~\cite{Danna}. To improve runtime, early hyperscale WANs therefore developed approximate methods~\cite{b4,swan}. Since then, there have been several major lines of work to improve the runtime-optimality tradeoff. We discuss two lines of work in depth -- solving the max-min fairness objective faster (\S\ref{sec:motivation-existing-maxmin}), and parallelizing solvers regardless of the objective (\S\ref{sec:motivation-parallelization}) -- which are relevant to our goals and methods and which we'll later compare with experimentally.

\subsubsection{Accelerating Max-Min Fairness}
\label{sec:motivation-existing-maxmin}

Compared to objectives like maximizing total throughput or minimizing max link utilization, \emph{max-min fairness} is more difficult because it involves solving multiple LP instances: first maximizing the allocation of the minimum flow, then maximizing the minimum of the remaining flows, and so on. Soroush~\cite{soroush} tackled this by using geometric binning of rates which then allows it to reformulate the problem as a single-shot optimization. This introduces an approximation, but one of the methods of~\cite{soroush} provides a configurable bound on suboptimality (dependent on the bin separation factor) with runtime increasing as the approximation improves.

\subsubsection{Parallelization Without Guarantees}

\label{sec:motivation-parallelization}
With problem sizes scaling, one would of course like to be able to parallelize the solver. This is tempting because large networks have many nodes, links, and flows which appear not to directly affect each other -- yet it is subtly difficult because in the global solution, they affect each other indirectly through overlapping shared resources.

Several recent works have tried to make progress towards parallelization, with different methods of splitting up the problem into smaller sub-problems. NCFlow~\cite{ncflow} splits a network into clusters (ideally capturing natural clustering in the topology), aggregating the problem at cluster granularity, and then solving parallelizable sub-problems within each cluster. However, this is in general not optimal for the global TE problem. POP~\cite{pop} partitions the problem differently, by splitting the set of user demands and slicing resources, so that each sub-problem has a subset of users and a fraction of network resources. This is a good fit for ``granular''~\cite{pop} problems, with many individually-tiny demands which can flexibly be shifted to use different resources. But that property is only partially true of TE, and after merging the sub-problems the solution is again not optimal in general.

Teal~\cite{teal} parallelizes across flow demands, applying a neural network model to route each individual demand. The use of ML makes it a natural fit to solve the sub-problems in parallel on a GPU. More relevantly for our work, after the RL stage, Teal uses a second stage optimization to improve the solution. This stage uses a \osrdel{Lagrangian }decomposition technique, specifically the Alternating Direction Method of Multipliers (ADMM)~\cite{neal2011distributed} which it can parallelize on a GPU. While ADMM in general can provide convergence guarantees, to make this particular TE formulation work, Teal avoids TE's non-negativity constraints, instead projecting the solution back into the feasible space. This results in approximations in each iteration of ADMM, which adds up over time, and the solution can (and does) converge to a suboptimal solution, without a bound on the optimality gap. Teal runs ADMM a fixed number of iterations to heuristically improve. \cite{teal} reports that this ADMM formulation did not converge fast enough to run alone, without the RL stage.

All the above approaches run into the difficulty that TE is not easily partitioned, and resort to heuristics which, for one reason or another, are not guaranteed to be optimal.

\subsubsection{Parallelization With Guarantees}

\osr{There is a large body of work on decomposing convex optimization problems into smaller subproblems. Examples include Benders\cite{benders1962partitioning}, Dantzig-Wolfe\cite{dantzig1960decomposition}, and dual decomposition\cite{guignard1987lagrangean, palomar2006tutorial}.} To the best of our knowledge, interestingly, all proposed parallel TE solvers with optimality guarantees leverage \osrdel{Lagrangean}\osr{dual} decomposition.\rb{@Brighten: does this positioning address (1) in your comment below?}

\osr{Another such family of techniques is the Augmented-Lagrangian splitting methods~\cite{bertsekas1989parallel, neal2011distributed}, that apply Lagrangian methods (for example, Lagrangian relaxation, Method of Multipliers) to existing decomposition techniques.} \bgc{This is basically good, but some nits: (1) The flow of the text is odd, because you mention ``Dual decomposition'' first in the prior paragraph without having defined it. (2) Dual decomposition should be written without capitalization, unlike Benders and Dantzig-Wolfe which come from people's names. (3) You can delete the sentence saying GATE uses ADMM; we've already mentioned that in the intro, and it's a bit out of place here because we're talking about past work.}

The above methods relate to general optimization problems, but there is a long history of applying decomposition theory to network optimization. Early works designed or analyzed congestion control algorithms \cite{kelly1998rate,low1999optimization,srikant2004mathematics,srikant2014communication}, where a user adapts its rate according to network feedback about Lagrange multipliers (acting as congestion signals). These are conceptually related to TE in that they are distributed solutions to optimally route traffic \cite{he2007towards} but do not explicitly solve WAN TE.

Closer to our problem, \cite{gpgpu-routing-globecom,WANG20181} solve related problems of path-selection and routing cost minimization. While these solutions also leverage \emph{decomposition to enable GPU parallelization}, the TE problem poses a very different set of constraints -- (i) a complex objective function and (ii) runtime-optimality tradeoff, that makes the formulation complex.

\bgnew{The closest works to ours are \cite{lin2006utility,zhang2025solving,DeDe} and we discuss each of these in more detail.}
\cite{lin2006utility} provides a formulation for multi-path TE using \textit{primal-dual descent} -- a decomposition technique that is an alternative to ADMM. However, this approach (i) requires careful manual tuning to make progress~\cite{neal2011distributed}, making it infeasible for large networks, and (ii) does not support max-min fairness. PDMCF\cite{zhang2025solving} also uses the primal-dual technique to implement a GPU-parallel multi-commodity flow solution. The formulation solves a simpler problem than \textit{WAN TE}, since it does not concern itself with a fixed pathset, a constraint that is key to practical WAN TE \bgnew{(e.g., to control latency). Instead, it allows any flow to traverse any link. This formulation enables PDMCF to achieve an $O(N)$ speedup over traditional solvers.  While \cite{lin2006utility,zhang2025solving} do not address practical constraints, we briefly report and compare their performance with \sysname in \cref{s:evaluation}.}




Recently, DeDe~\cite{DeDe} comes even closer to the goal \bgnew{of parallelizing WAN TE with practical constraints}. It used an ADMM decomposition \rbnew{for WAN TE}, but solved the individual flow sub-problems with calls to a general-purpose LP solver (Gurobi), allowing it to preserve the problem's non-negativity constraints. As a result, for convex objective functions (such as maximizing total flow or minimizing max link utilization), DeDe converges to the optimal allocation.\footnote{While~\cite{DeDe} does not include a formal proof, its optimality should be inherited from its use of ADMM.} However, DeDe has two limitations that are relevant to our goals. First, it does not support solving for the max-min fairness objective commonly used in deployed TE systems. Second, while use of a general-purpose solver as a subroutine allows DeDe to flexibly apply to many optimization problems even outside of TE, it comes at a relatively high efficiency cost, due to: (a) the overhead of hundreds of calls to Gurobi, (b) the need to run general-purpose solvers like Gurobi on CPU instead of GPU, and (c) the many iterations for ADMM to converge.
\begin{figure*}[!t] 
    \centering
    \includegraphics[width=0.99\textwidth]{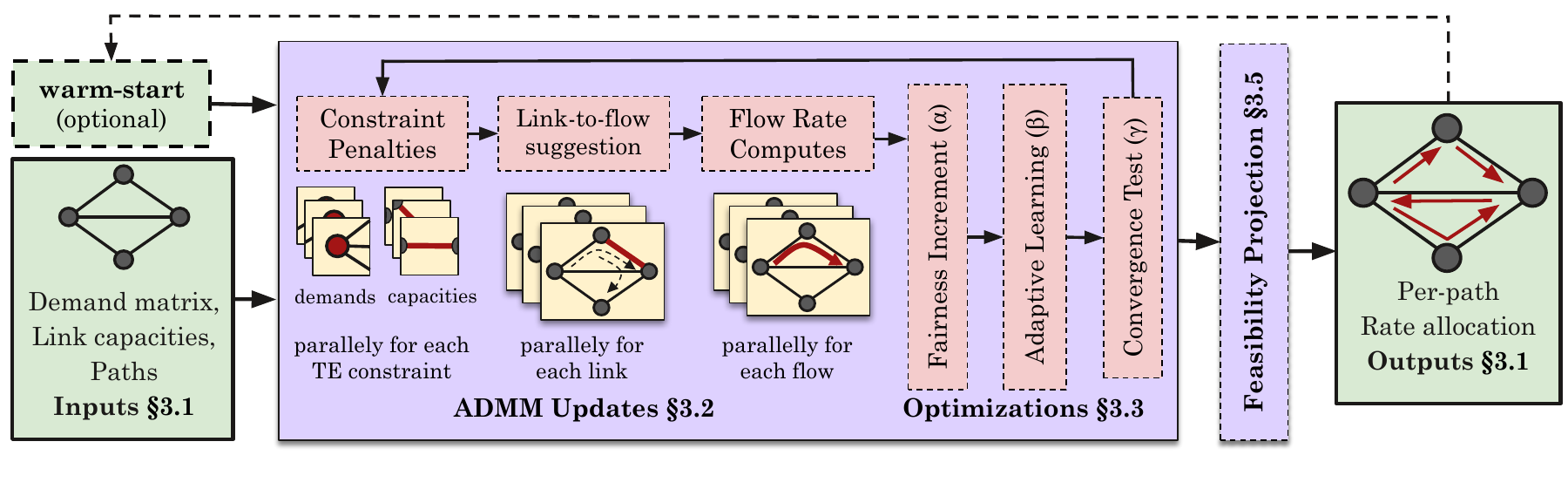}
    \caption{\osrdel{Steps in one iteration of \sysname. All \textit{blocks} within a step are computed parallelly. Each steps is computed sequentially} \osr{\sysname's design: Inputs (demands, capacities, paths) are fed to \sysname's iterative loop that recomputes rates as they progress to the optimal. Each iteration involves (i) ADMM updates (\cref{s:design:iterative}) and (ii) convergence optimizations (\cref{s:design:convergence-optimizations}), all of which solve parallel subproblems on a GPU. The output of the iterations is fed to Feasibility Projection (\cref{s:design:projection}) to resolve constraint violations from the final allocations.}}
    \label{fig:iteration}
\end{figure*}
While these works have pushed the boundaries, they imply two open questions: Is it possible to parallelize TE in a principled way that is guaranteed to converge to the optimal, with practical constraints and objectives (e.g., path-based, max-min fairness)? Furthermore, can this be done with a direct formulation that does not resort to expensive calls to a general-purpose optimization package? As we will see, we can answer these questions in the affirmative. As in the recent work above, \osrdel{Lagrangian decomposition}\osr{leveraging a decomposition technique} forms our principled foundation. But to answer the open questions, we \osrdel{need to }
introduce a novel carefully-designed closed-form formulation of TE, combined with an effective and theoretically-sound solution to ADMM's slow convergence problem, and a method of handling the complex nonlinear nature of max-min fairness -- allowing us to leverage GPUs and \osrdel{continual re-optimization }\osr{warm-start}
for \osrdel{huge}\osr{significant} performance gains.

%% file: 3-design.tex
\section{Design}\label{s:design}

\begin{table}[h]
    \centering
    
    \rowcolors{1}{gray!30}{gray!10} 
    
    \begin{tabular}{lr}
        \toprule
        \hline
        \textbf{Term} & \textbf{Interpretation} \\\hline 
        $E$ & Set of all links in the network. \\\hline
        $N$ & Set of all nodes in the network. \\\hline
        $P$ & Set of all allowed paths. \\\hline
        $P_{st}$ & Set of allowed paths from node $s$ to $t$.\\\hline
        $C_{e}$  & Capacity of link $e$. \\\hline
        $D_{st}$ & Demand requested from node $s$ to $t$. \\\hline
        $x_{r}$ & Rate allocated over path $r$. \\\hline
        $S_{st}$ & Total rate allocated from node $s$ to $t$ (over $P_{st}$).\\\hline\bottomrule
    \end{tabular}
    \caption{Notation for a TE formulation.}
    \label{t:notation-te}
\end{table}

\begin{table}[t]
    \centering
    \renewcommand{\arraystretch}{1.1}
    
    \rowcolors{1}{gray!30}{gray!10} 
    
    \begin{tabular}{lr}
        \toprule
        \hline
        \textbf{Term} & \textbf{Interpretation} \\ \hline
        $y_{e, r}$ & Suggested value of $x_{r}$ by edge $e$. \\ \hline
        $dual_{dem_{st}},dual_{cap_{e}},$ & Duals for primal constraints,\\
        $dual_{neg_{r}}, dual_{con_{e,r}}$ & (\cref{eq:demand-constraint-primal,eq:capacity-constraint-reformulated,eq:non-negative-primal,eq:consensus-reformulated-true})\\\hline
        $\beta$ & Learning rate. \\ \hline
        $k$ & Iteration index. \\\hline\bottomrule

    \end{tabular}
    \caption{Additional terms for \sysname's ADMM re-formulation.}
    \label{t:notation-gate}
\end{table}

We aim to design an algorithm to meet the needs discussed in \cref{sec:motivation}; namely, providing both fast reaction time and a principled \bgnew{optimality guarantee}. This motivates an algorithm that makes useful progress in a few iterations,
can be warm-started from the previous allocation, and provably converges to the true optimum if allowed to run. Before presenting our algorithm design to meet these goals, we first formally state the TE problem.
\begin{figure}[ht]
    \centering
    \includegraphics[width=\linewidth]{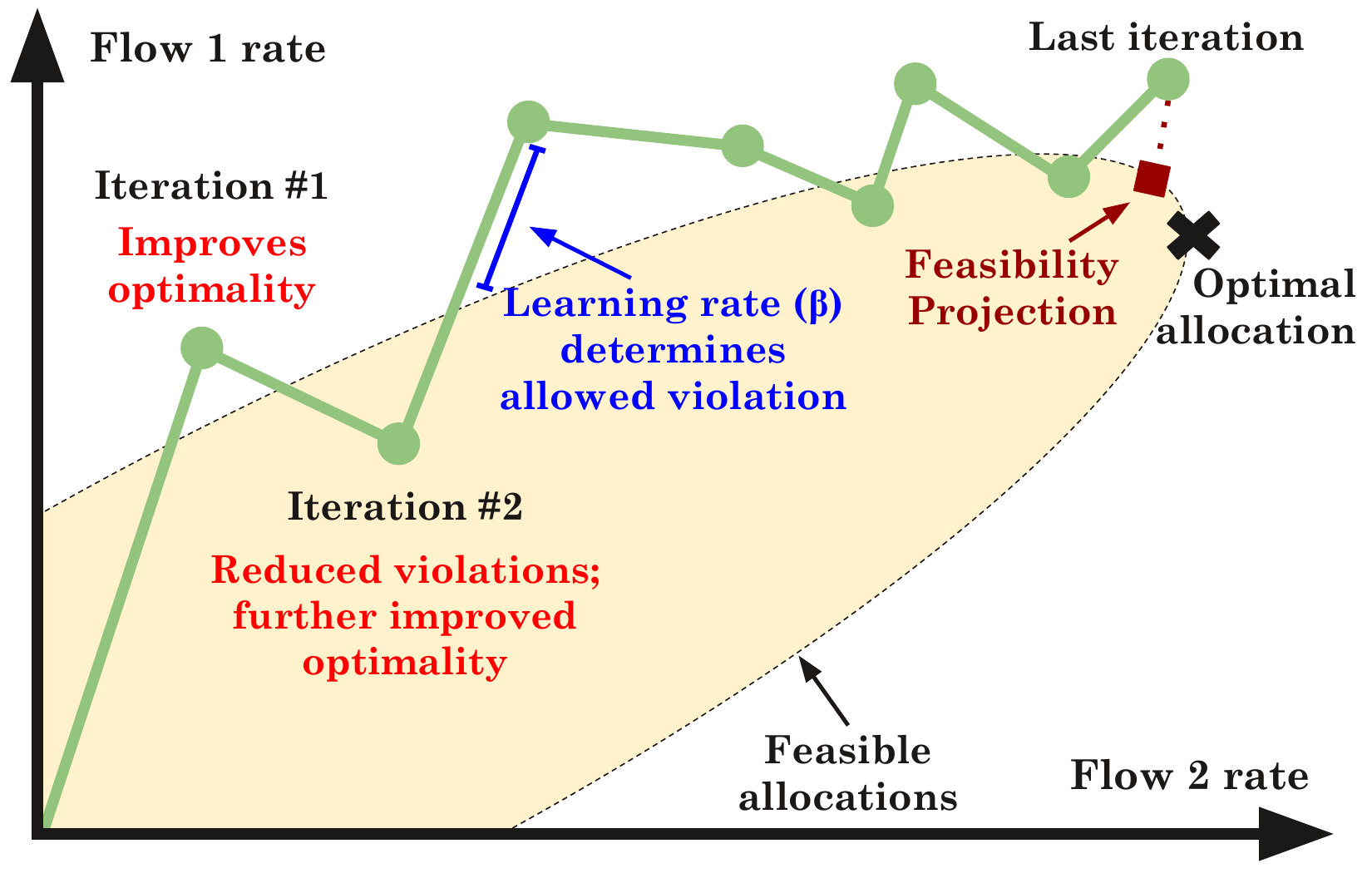}
    \caption{\osr{Toy example of a network with two flows, describing how \sysname iteratively progresses towards the optimal allocation. In each iteration, \sysname finds the best allocation to \textit{jump} to from the previous iteration, while balancing constraint violations and optimality. \sysname concludes iterations when progress slows  $(\gamma)$, and projects the allocation to a feasible configuration.}}
    \label{fig:gate-cartoon}
\end{figure}
\subsection{TE Formulation}\label{s:design:problem}
\osrdel{WANs typically perform updates in two phases. First, there is a path computation phase}\osr{Deciding how a traffic demand should be routed through a network  typically involves two steps. The first is a path computation phase} which computes a set of allowed paths (e.g., k-shortest paths) for each node-pair. \osrdel{Then, these allowed paths are installed into the routers' forwarding tables. The TE algorithm's job is then to find the optimal rates over each path.}\osr{The second is to determine how a traffic demand should be split over this set of allowed paths. This latter is the job of the TE solver which computes the ideal rates for each path. The final outcomes are then programmed into routers.} \osrdel{As with prior TE \osrdel{solutions}\osr{solvers}~\cite{swan, DeDe, soroush,ncflow,pop}, we assume that the set of paths to be used are precomputed and provided as input ($P$).}

\paragraph{\textbf{Variables}}
We summarize the notations used in the TE formulation in Table~\ref{t:notation-te}. \osrdel{A node-pair $(s,t)$ (which can be interpreted as a ``user'') requests a certain demand $d_{st}$ from node $s$ to node $t$.}

\paragraph{\textbf{Inputs and Outputs}}
TE algorithms expect three inputs: (i) link capacities, (ii) requested demands, and (iii) a set $P$ of pre-computed allowed paths. \bgc{I believe it would be more precise and consistent with the notation below to write: (iii) a set $P_{st}$ of pre-computed allowed paths for each $(s,t)$. ...right? Are you additionally defining $P = \bigcup_{s,t} P_{st}$ ?} The TE algorithm will output the rate ($x_{r}$) for all $r \in P$.

\paragraph{\textbf{Constraints}}
We constrain rates to (i) not exceed requested demands (Eq. \ref{eq:demand-constraint-primal}), (ii) not exceed link capacities (Eq. \ref{eq:capacity-constraint-primal}) and (iii) non-negative values. (Eq. \ref{eq:non-negative-primal}):
\begin{equation}
    \sum_{r}^{P_{st}} x_{r} \le D_{st} \qquad \forall s \in N, t \in N
    \label{eq:demand-constraint-primal}
\end{equation}
\vspace{-2em}
\begin{equation}
    \sum_{r:e \in r}^{P} x_{r} \le C_{e} \qquad \forall e \in E
    \label{eq:capacity-constraint-primal}
\end{equation}
\vspace{-1em}
\begin{equation}
    x_{r} \ge 0 \qquad \forall r \in P
    \label{eq:non-negative-primal}
\end{equation}

\bgc{I think the first equation sum would be more clearly written as $\sum_{r\in P_{st}}$ instead of $\sum_{r}^{P_{st}}$ but that may be a personal preference. Also, related to the discussion above, ensure $P_{st}$ and $P$ are defined.}

\paragraph{\textbf{Objective}}
TE objectives can vary by operator, with max-min fairness being the most commonly deployed in production networks~\cite{swan, b4}, and recent \bgnew{research proposals} for maximum multicommodity flow~\cite{ncflow, pop, DeDe, teal}, \bgnew{which we also refer to as simply ``max-flow'' for short here.}  \sysname's goal is to be generalizable and effective for a spectrum of objectives. In that spirit, we use the standard notion of $\alpha$-fair utility to compute the utility for a node-pair ($U_{\alpha}$), and aggregate the utilities across node-pairs to compute the objective ($U^{*}$):

\begin{equation}
\begin{aligned}
    U_\alpha(b) = \dfrac{(b)^{(1-\alpha)} - 1}{1-\alpha}\\\\
    U^{*}(x_*) = \sum_{s,t}^{N, N} U_{\alpha}\left(\sum_{r}^{P_{st}}x_{r}\right)
\end{aligned}
    \label{eq:objective-primal}
\end{equation}

This captures a spectrum of objectives with varying $\alpha$. Table~\ref{t:alpha-objectives} lists the objectives of particular interest. 

\begin{center}
    \vspace{1em} 
    \centering
    \renewcommand{\arraystretch}{1.3}
    \rowcolors{1}{gray!30}{gray!10} 
    
    \begin{tabular}{lr}
        \textbf{$\bm{\alpha}$} & \textbf{Objective} \\ 
        $\alpha = 0$ & Max-flow \\
        $\alpha \rightarrow 1$ & Proportional fairness ($log$) \\
        $\alpha \rightarrow \infty$ & Max-min fairness \\ 
    \end{tabular}
    
    \captionof{table}{Objectives captured by $\alpha$-utility.}
    \label{t:alpha-objectives}
    \vspace{1em} 
\end{center}

\paragraph{\osr{\textbf{Rewriting TE into a \textit{decomposable} formulation}}}
\osr{\cref{eq:demand-constraint-primal,eq:capacity-constraint-primal,eq:non-negative-primal,eq:objective-primal} formulate the TE problem compactly, but they are coupled in a way that is not decomposable. A single variable $x_r$ appears in the objective, the demand constraint, and the capacity constraint for every link along $r$. A TE solver that uses this formulation must reason about all these constraints together. \sysname therefore rewrites it, in two steps, into an equivalent problem that decomposes.}

\osr{First, we \emph{split the coupled variables.} We give each link its own private copy of every rate $x_r$ that crosses it, with an additional constraint that ensures the copy has the same value as $x_r$:}

\begin{equation}
    \sum_{r:e \in r}^{P} y_{e, r} \le C_{e} \qquad \forall e \in E
    \label{eq:capacity-constraint-reformulated}
\end{equation}
\vspace{-0.5em}
\begin{equation}
    x_{r} = y_{e, r} \qquad \forall e\in r, r \in P
    \label{eq:consensus-reformulated-true}
\end{equation}
\osr{This is the same TE problem -- any solution of one is a solution of the other -- but it has a property the original lacks: \textit{every variable now faces exactly one constraint}. Link $e$ owns the copies $y_{e,\cdot}$ and worries only about its own capacity; path $r$ owns $x_r$ and worries only about its demand. The global coupling that tied them together has been concentrated into \cref{eq:consensus-reformulated-true} -- a single requirement that the two sides should agree.}
\osrdel{\sysname aims to break the global TE problem into smaller, \textit{local} sub-problems, which still converge to the \textit{global} optimum when solved iteratively. \sysname achieves this by applying Lagrangian decomposition techniques. To do so, we first transform \cref{eq:capacity-constraint-primal} to introduce auxiliary variables,}

\osrdel{Note that this transformation solves the same TE problem, but separates constraints such that there is exactly one constraint per variable. This transformation enables each variable to iteratively solve its own local constraint, while exchanging information with other variables across iterations.}

\osr{Second, \emph{we penalize violations instead of forbidding them}. Splitting the TE problem allowed each variable to be constrained by one entity, but the problem is still globally constrained. We now replace the constrained problem with a sequence of unconstrained problems which will eventually converge to the same constrained problem. We form the standard Augmented Lagrangian\cite{auglag} of the split problem, folding the objective and all four constraint families -- demand (\cref{eq:demand-constraint-primal}), capacity (\cref{eq:capacity-constraint-reformulated}), non-negativity (\cref{eq:non-negative-primal}), and consensus (\cref{eq:consensus-reformulated-true}), into a single unconstrained expression of the form:}\osrdel{Next, we compute the standard Augmented Lagrangian~\cite{auglag} for this transformed formulation, by introducing dual variables $dual_*$ and a penalty constant $\beta$.}
\osr{
\begin{equation*}
  \mathcal{L}_\beta = \underbrace{-\,U^{*}(x_*)}_{\text{primal objective}}
  \;+\; \frac{\beta}{2}\sum_{\text{constraints } c}
  \big(\underbrace{dual_c + \text{residual}_c}_{\text{constraint penalty}}\big)^2 ,
\end{equation*}}

\osr{\noindent where $\text{residual}_c$ is the amount of constraint $c$ that is currently violated, $dual_c$ tracks how much it has been violated in the past instances, and $\beta$ sets how expensive a violation is relative to the objective. \cref{s:appendix:extended-design:lagrangian} gives $\mathcal{L}_\beta$ in full.}

\osrdel{\begin{equation}
\begin{aligned}
    & \mathcal{L}_\beta(x_*,y_*;dual_{*}) = - \underbrace{U^{*}(x_*)}_{\text{primal objective}} \\
    &\quad + \frac{\beta}{2}\sum_{(s,t)} \Big(\underbrace{dual_{dem_{st}} + (\textstyle\sum_{r\in P_{st}} x_r - D_{st})}_{\text{demand constraint residual}}\Big)^2\\
    &\quad + \frac{\beta}{2}\sum_{e\in E} \Big(\underbrace{dual_{cap_e} + \textstyle(\sum_{r:\,e\in r}y_{e,r} - C_e)}_{\text{capacity constraint residual}}\Big)^2\\
    &\quad - \frac{\beta}{2}\sum_{r\in P} \Big(\underbrace{dual_{neg_r} - x_r}_{\text{non-negativity constraint residual}}\Big)^2\\
    &\quad + \frac{\beta}{2}\sum_{e\in E}\sum_{r:\,e\in r} (\underbrace{dual_{con_{e,r}} + (x_r - y_{e,r})}_{\text{consensus constraint residual}})^2
\end{aligned}
\label{eq:augmented-lagrangian}
\end{equation}}

\osr{This reformulation has two consequences: First, minimizing $\mathcal{L}_\beta$ achieves the optimal solution of the original TE problem at a feasible point where every penalty vanishes (with a large $\beta$), leaving only $U^*$. \bgc{Is that true and obvious? Doesn't it depend on some things like $\beta$ being nonzero and sufficiently large for the constraint violation to be not worth the cost?} Second, an allocation that violates constraints is not rejected outright, but is penalized for it. This lets \sysname move through infeasible allocations on its way to the \emph{feasible, optimal allocation}. As we see in \cref{s:design:iterative}, this reformulation matches the template needed to apply the Alternating Direction Method of Multipliers (ADMM)~\cite{neal2011distributed} technique.}

\osrdel{Note that the Augmented Lagrangian folds the objective and constraints from \cref{s:design:problem} into a single expression, by representing the constraint violations as \textit{residuals} and penalizing them. It intuitively follows that the optimal $x_*$ is the same for \cref{eq:augmented-lagrangian} and \cref{eq:objective-primal}.}


\subsection{\sysname's ADMM Iterations}
\label{s:design:iterative}
\osrdel{Next, we use ADMM decomposition~\cite{neal2011distributed} to transform the Augmented Lagrangian into an iterative formulation. The goal of each iteration is to compute values of all variables ($dual_*, x_*, y_*$) such that the Lagrangian is maximized -- computed by setting $\frac{\delta \mathcal{L}}{\delta var} = 0$. Within each iteration, a series of sequential \textit{steps} compute these values, with information from the previous iteration. Done iteratively, this eventually converges to the global optimal. Each step involves solving a large number of small subproblems, all of which are solved in parallel, giving \sysname massive parallelization benefits (illustrated in \cref{fig:iteration}). The steps within one iteration are: (i) Dual updates (Ln 3), (ii) Rate suggestion updates (Ln 3), (iii) Rate updates (Ln 4), and (iv) convergence optimizations (Ln 5-15). Note that the Lagrangian is improved in two ways - (i) improving the primal objective and, (ii) reducing the constraint violation penalties, both of which are desirable in computing the optimal TE allocation.}

\osr{Next, \sysname applies ADMM~\cite{neal2011distributed} to turn $\mathcal{L}_\beta$ into an iterative procedure. In each iteration, the rates (and its copies) are assigned a \textit{locally-optimal} value that minimizes $\mathcal{L}_\beta$, computed by setting $\frac{\delta \mathcal{L}}{\delta var} = 0$. Repeated a few iterations, \bgc{Is that hiding anything? I wonder if we need a more nuanced/qualified statement. I thought we might need a lot of iterations and then may not reach the exact optimal in practice} this converges the rates to their optimal allocation (\cref{s:design:opt-guarantee}). Each ADMM iteration has three steps: (i) updating the penalties for constraint violations, (ii) computing a suggested rate from the perspective of every link a flow traverses, \bgc{That phrasing ``by every link'' may be confusing. Do we mean we are thinking of each link as a local subproblem with the link taking action to compute its own suggested rate for each flow? If so, that will not be clear to many readers.  Could say: ``from the perspective of each link it traverses''} (iii) (re-)computing the rate of each flow. Each step is a large bank of independent subproblems solved in parallel. Within a step, no subproblem reads another's output, so all instances run concurrently; between steps, the only communication is the values computed by the previous step. Two things follow. The subproblems are small enough to sit on a single GPU thread, and there are more of them as the network grows -- so parallelism improves with scale rather than degrading with it, which is the opposite of how TE solvers typically behave (\cref{s:evaluation:scaling-behavior}).}\osrdel{\paragraph{\textbf{Initialization}}
We initialize $x_{*}$ with an optional input. This enables \sysname to \textit{warm-start} making progress with a near-optimal allocation, if available. Otherwise, our default is to start by uniformly spreading the rate across available paths $x_r = \frac{D_{st}}{|P_{st}|}$. Note that this may violate some primal constraints, but constraint violations get resolved over iterations. We initialize suggested rates $y_{e, r} = x_{r}$, dual variables $dual_* = 0$ and slack variables $slack_{*} = 0^+$.}
\osr{\paragraph{\textbf{Initialization}}
At initialization, \sysname sets (i) $dual^0_*=0$, since there is no past information of constraint violations, and (ii) $y^0_{e,r}=x^0_r$ to get each link copy to agree with the rate, and (iii) the initial allocation $x^0$ by either a fair-share of the demands $\frac{D_{st}}{|P_{st}|}$, or \textit{warm-start} with a recent allocation. This may violate constraints, which is permitted since violations are penalized, not rejected (\cref{s:design:problem}), and are resolved over the iterations.}
\osr{\paragraph{\textbf{Warm-Start}}
\rb{Sylvia: you had a few comments on warm-start's old version. This one tries to address it. Please review if you get a chance}
Because \sysname progresses towards the optimal through a sequence of \textit{locally-optimal} moves, seeding the rates $x^0$ with an allocation already close to the optimum \emph{is likely} to require fewer jumps. This makes warm-starting a natural choice, instead of recomputing the entire allocation from scratch. This is of value since consecutive demand matrices resemble each other even across sizeable swings, so a recent allocation can be a good starting point. Warm-start is known to help ADMM-based algorithms in other settings~\cite{neal2011distributed}; \sysname applies it to TE and we quantify the benefits in \cref{s:evaluation:factoranalysis}.}

\paragraph{\textbf{Constraint Penalty Computation}}
\osr{First, we update the dual variables $dual_c$, one per constraint $c$ in $\mathcal{L}_\beta$. The dual variable capture the running total of a constraint's violation: it grows while the constraint is violated and decays back toward zero when it is satisfied. It is the only signal by which one subproblem learns what the others have been doing. All four families of constraints -- demand, capacity, non-negativity, and consensus -- share the same form, $dual \leftarrow \left[dual + \text{residual}\right]_+$. For example, the capacity dual is updated as,}
\begin{equation}
\label{eq:dual-updates-cap}
dual_{cap_{e}}^{k+1}= \left[dual_{cap_{e}}^{k}+\Big(\sum_{r:e\in r}^{P}y_{e,r}^{k}-C_{e}\Big)\right]_+
\end{equation}
\osrdel{The dual variables represent the constraint violations of the TE problem.  They act as signals for the subsequent steps to update rates such that constraint violations are reduced. \cref{eq:dual-updates-cap,eq:dual-updates-dem,eq:dual-updates-con,eq:dual-updates-neg} have $|D|$, $|E|$, $|P|*|E|$, and $|P|$ instances, all computed in parallel. Adapting TE to the ADMM formulation requires additional \textit{slack} variables, described in \cref{s:appendix:extended-design:slack}.}
\osr{\noindent A link whose capacity is exceeded by the flows traversing $y_{e,r}$ sum to more than its capacity accumulates a positive $dual_{cap_e}$, and the longer it has been over-capacity, the larger that value grows. The remaining three updates, and some additional slack variables that let \sysname treat the inequality constraints as equalities, are given in \cref{s:appendix:extended-design:slack}. In total this step computes $|E|$ capacity duals, $|D|$ demand duals, $|P|\!\cdot\!|E|$ consensus duals, and $|P|$ non-negativity duals, none of which depends on any other.}

\paragraph{\textbf{Suggested Rate Computation}} 
\osr{Next, we want to compute \textit{suggested} rates $y_{e,r}$ from a link's perspective, one for every link a flow traverses. For each link, we would score flows traversing that link, and redistribute bandwidth around the link's own average. The score $cost_{e,r}$ combines (i) what the flow is currently sending $x_r^{k-1}$, (ii) the link's accumulated congestion penalty (dual), and (iii) how far that path has drifted from the link's last suggestion. Then, a path cheaper than the link's average cost is told it may take more, one that is more expensive is told to back off:}
\osrdel{The dual variables enable quantifying constraint violations. Next, we need to compute rates that  minimize constraint violations. \bgc{Is it really only to minimize constraint violations or also to move toward optimality?} Since we decompose this problem into independent subproblems, we first look at the subproblems tasked with enforcing capacity constraints. Each edge $e$ computes a ``suggested'' rate $y_{r}$ for every flow passing through it $x_{r}$. If the link is experiencing congestion (high $dual_{cap}$), it would ``suggest'' flows to back off. While deciding the extent to \textit{back off} for each flow, it would choose values such that no flow deviates disproportionately from their current rates $x_r$. The mathematical representation of these two considerations is reflected in \cref{eq:penalty-b}, and the rate suggestions are computed in \cref{eq:rate-suggestion}, where $n_e$ is the number of paths traversing edge $e$. All $|E| * |P|$ instances of $y^{k+1}$ are computed in parallel. \srnew{minor: run the above para through AI for minor grammatical adjustments.}}

 


\begin{equation}
    cost_{e,r} = \beta \left(-C_e - x_r^k\right) - dual_{cap_e}^{k+1} - dual_{con_{e,r}}^{k+1}
    \label{eq:penalty-b}
\end{equation}

\begin{equation}
    y_{e,r}^{k+1} = \left[ \frac{\bar{cost}_e - cost_{e,r}}{\beta} \right]_+ \quad \text{where} \quad \bar{cost}_e = \frac{1}{n_e} \sum_{r: e \in r} cost_{e,r}
    \label{eq:rate-suggestion}
\end{equation}



\paragraph{\textbf{\osr{Flow Rate Computation}\osrdel{Rate Updates}}}
\osrdel{With the duals ($dual_*$), slack variables ($slack_*$), and rate suggestions ($y_{e,r}$) computed, we extract the optimal path rates. This \textit{local} update seeks to (i) optimize the objective \cref{eq:objective-primal}, (ii) minimize demand violations \cref{eq:demand-constraint-primal}, and (iii) maintain consensus with the link-level rate suggestions \cref{eq:consensus-reformulated-true}. }
\osr{Finally, the flow rate is computed, reconciling the rate suggestions with what the objective wants. Prior decomposition techniques hand this off to an LP solver~\cite{DeDe} every iteration. Instead, \sysname transforms this LP into a closed form update, to allow the entire loop to stay on the GPU.}

\underline{Sum computation:} Because the utility objective and demand constraints operate on the aggregate flow across all its paths, this step first requires computing the optimal $S_{st} = \sum_{r \in P_{st}} x_{r}$, even before we can compute the individual $x_r$. What we have is an equation of the form $AS^{\alpha - 1} + BS + C = 0$, where $\alpha$ depends on the objective (\cref{eq:objective-primal}). This makes the problem significantly complex -- and we compute this using the parallel Newton-Bisection method~\cite{stoer2002numerical}, detailed in \cref{s:appendix:extended-design:sum-computation}. Next, we use $S_{st}$ back to compute $x_{r}^{k+1}$ as the following,


\vspace{-0.5em}
\begin{equation}
    \boxed{
    \begin{aligned}
        x_r^{k+1} = \frac{1}{(1+n_r)} \bigg[ 
        &\underbrace{\sum_{e\in r}(y_{e,r}^k-dual_{con_{e,r}}^k)}_{\text{Consensus with $y_r$}} + 
        \underbrace{\frac{(S_{st})^{-\alpha}}{\beta}}_{\text{Objective}} \\
        &- \underbrace{\big(S_{st}-D_{st} + dual_{dem_{st}}^k\big)}_{\text{Demand constraint violation}} \\ &+ \underbrace{dual_{neg_r}^k}_{\text{Non-negativity}} \bigg]
    \end{aligned}
    }
    \label{eq:rate-computation}
\end{equation}
\osrdel{which accounts for (i) objective function, (ii) non-negativity constraint and, (iii) demand constraints.}

\osr{\noindent \cref{eq:rate-computation} is worth reading as four factors being accounted together -- (a) pulling the rate toward the suggestions for every link it crossed, (b) a term capturing the objective, which desires that the flow sends more (and fairly), (c) pulling the rate down when the node-pair is over-serving the demand, and (d) pulling the rate up if it becomes negative. The final rate computation update is their average, weighted by $n_r$, the number of links on the path.}


\osrdel{\paragraph{\textbf{Warm start}}\osr{
Rate updates iteratively converge to the global optimal by (i) making locally-optimal updates, and (ii) reacting to the rate suggestion updates from each link. Therefore, \sysname progresses toward the global optimum via a series of localized jumps. It follows that seeding the initial rates ($x^0$) with values closer to the global optimum \textit{is likely} to \srnew{why the emph on "can"? (More generally, is this a guarantee? If yes just say: "will require fewer jumps". If no, maybe say "is likely to require fewer \ldots"}\rb{It is not a guarantee, was emphasizing on can to indicate that} require fewer jumps to reach the global optimum. The learning rate $\beta$ enables a complementary axis of optimization, i.e. the \textit{size} of each jump, discussed in \cref{s:design:iterative}.
Warm-start has been studied with decomposition techniques applied to other contexts\cite{neal2011distributed}. \srnew{Were these in the TE context? If not, say "ADMM based algorithms applied to other contexts." (To make our contribution clear.)}\rb{Done}. In TE, this is particularly useful since the network traffic holds resemblance to history, despite the rapid swings, thus enabling the stale solutions from recent past to be useful points to warm-start from. We show empirical results of the effectiveness of warm start in \cref{s:evaluation:factoranalysis}.
} \srnew{Last couple sentences are awkward. Run through AI for better phrasing?}}

\srnew{Do we reference Figure 2 anywhere in the text? If not, this might be a good place.}

\srnew{Ditto Figure 1?}

\subsection{Accelerating Convergence} 
\label{s:design:convergence-optimizations}
\srnew{Maybe call this "Accelerating Convergence through Adaptive Learning" (Up-leveling the techniques so it doesn't feel buried as an afterthought optimization.) And delete para title.\rb{Not sure I am following? We also have alpha increments and convergence test (which is somewhat oddly placed) in the subsection}.
\bg{Could we call this something like ``Optimizing the Convergence Process''? ``Tuning Updates'' is a bit hard to parse and ``tuning'' makes it sound a bit trivial}}
\osr{Naively applied, \sysname's ADMM iterations converge slowly and are highly sensitive to the choice of parameter $\beta$. We describe our techniques that are responsible for making \sysname's decomposition practical for TE.}
\label{sec:optimizing-convergence}
\paragraph{\textbf{Adaptive Learning}}

Performance of the iterative algorithm can be highly sensitive to the choice of $\beta$. It determines how quickly the iterations converge, and how close to optimal they converge. How do we choose a good learning rate? In each iteration, the rate updates $x_r$ attempt to maximize the Lagrangian $\mathcal{L}$ (\cref{eq:augmented-lagrangian}). It follows that a high $\beta$ would heavily penalize constraint violations such that the rate updates prioritize ``being feasible'' over ``finding the best allocation''. The algorithm is effectively ``stiff'', i.e. it satisfies limits but struggles to move towards the optimal solution ($U^*$). We capture the stagnation in rates using the \textit{dual residual} ($s^{k+1}$). It is defined as 
\begin{equation}
s^{k+1} = \| x_*^{k+1} - x_*^k \|
\label{eq:dual-residual}
\end{equation}

Conversely, if $\beta$ is too low, the rate updates place more weight on maximizing the objective $U^*$ than on minimizing constraint violations. While this yields a high objective value quickly, the resulting rates violate capacity or demand constraints, meaning they cannot physically fit on the network. We compute constraint violations using the \textit{primal residual} ($r^{k+1}$),
\begin{equation}
r^{k+1} = \|\text{dual}_{*}^{k+1} - \text{dual}_{*}^k\|
\label{eq:primal-residual}
\end{equation}

Our goal is for \osrdel{the iterates}\osr{\sysname's iterations}
to balance optimizing the objective with enforcing constraints. We achieve this by dynamically updating $\beta$ to keep the primal and dual residuals within a similar order of magnitude.


We implement this using a \textit{Residual Balancing} technique~\cite{neal2011distributed}. This scheme inflates $\beta$ by a factor $C$ when the primal residual is much larger compared to the dual residual (signaling excessive constraint violation), and deflates $\beta$ by $C$ when the dual residual dominates (signaling slow convergence). \cref{eq:adaptive-learning} defines the $\beta$ update rule, with $A$ and $C$ as tunable constants. In \cref{s:evaluation:factoranalysis}, we show that \bgnew{our adjustment scheme for $\beta$ roughly halves the number of iterations needed to converge.}

\begin{equation}
    \begin{aligned}
        \boxed{
\beta^{k+1} := \begin{cases} 
C* \beta^k & \text{if } r^{k+1} > A * s^{k+1} \\
\beta^k / C & \text{if } s^{k+1} > A * r^{k+1}\\
\beta^k & \text{otherwise}
\end{cases}
}
    \end{aligned}
    \label{eq:adaptive-learning}
\end{equation}

\paragraph{\textbf{$\alpha$-increments}}
Direct optimization of the $\alpha$-fair utility function (\cref{eq:objective-primal}) for large values of $\alpha$ introduces numerical stiffness. When $\alpha \gg 1$, the utility curve becomes exceedingly sharp, as shown in \cref{fig:x-v-alpha}. It experiences large gradient swings at the beginning, followed by a nearly flat curve, that slows convergence within the ADMM \bgnew{iterations}.

To address this, we borrow inspiration from the Homotopy Analysis Method\cite{liao2004homotopy} -- i.e., solve a simpler problem and then modify it to the solution of a more complex solution. Instead of directly solving for a large $\alpha_{target}$, we restructure the problem to begin at $\alpha = 0$ (a purely linear, maximum-throughput objective). Once the ADMM iterations converge for a given $\alpha$ -- meaning the primal and dual residuals fall below the threshold $\gamma$ -- we increment $\alpha$ and resume optimization. This is beneficial for two reasons -- (i) it is easier to converge for smaller $\alpha$, and (ii) for large $\alpha$, the optimal allocations for $\alpha$ and $\alpha + 1$ are geometrically adjacent in the feasible space (as seen in \cref{fig:alpha-v-opt} and \cref{fig:x-v-alpha}), leading to fewer iterations of the ADMM loop than starting from scratch.

These characteristics naturally offer a unique stopping condition when solving for max-min fairness ($\alpha_{target} \rightarrow \infty$). While the $\alpha$-increment loop is theoretically unbounded, \sysname continually checks if the network state has stagnated upon an $\alpha$ increment (i.e. $residuals \leq \gamma$ in the next iteration after the increment).
When that happens, \bg{(1) That's confusing. It sounds like you need to run one iteration to check that nothing is changing, not zero iterations. (2) FOr my understanding ... you are really using \emph{exact same} as the condition? So, it's an even stricter condition than the usual $\gamma$ stopping criterion?} we know that \sysname has converged (under $\gamma$) for all higher $\alpha$, including max-min fair.

In \cref{s:appendix:extended-design:optimality}, we describe how this provably converges to the max-min fair allocation.

\paragraph{\textbf{Convergence Criteria}}
\osr{Ideally, \sysname's iterations should terminate when the rate recomputes achieve (i) a near-optimal allocation, and (ii) low constraint violations.} \bgc{If I understand correclty, you are using ``recomputes'' as a noun, but it is actually a verb so it's hard to parse. The noun form would be ``recomputations''. There's a similar problem in Figure 1: I think ``Flow Rate Computes'' would be better written ``Flow Rate Computations''.} Similar to adaptive learning rates, we also use the primal- and dual-residuals for the stopping criterion for \sysname. We terminate iterations when both the primal residual (indicating constraint feasibility) and the dual residual (indicating solution stationarity) fall below a threshold $\gamma$ (\cref{eq:convergence}). A low $\gamma$ yields a more optimal solution but requires more iterations. We discussed in \cref{sec:motivation:runtime-optimality-tradeoff} that we want to achieve a good tradeoff between runtime and optimality. \srnew{What follows was confusing because it first says based on topology and traffic and then say that we tuned based on DAO which is a metric definition. Also, I thought operators would tune gamma based on their desired runtime-optimality tradeoff but that's not topology,traffic, or DAO?}  We expect $\gamma$ to be selected by the operator based on their network topology and traffic patterns. For example, in \cref{s:evaluation}, we \osrdel{select}\osr{tune} $\gamma$ to optimize for our metric of interest -- \textit{Drift Adjusted Optimality}. \bg{The phrasing there is a bit confusing -- ``we show that we select $\gamma$...'' -- do you mean in \cref{s:evaluation} we show \emph{how} to select $\gamma$?}

\begin{equation}
r^{k+1} \leq \gamma \quad \text{and} \quad s^{k+1} \leq \gamma
\label{eq:convergence}
\end{equation}

\subsection{\osrdel{Resolving Constraint Violations}\osr{Feasibility Projection}}\label{s:design:projection}
\rbnew{
While \osrdel{\sysname's subproblems}\osr{the \sysname iterations
in \cref{s:design:iterative} and \cref{s:design:convergence-optimizations}} 
provably converge to \osrdel{a}\osr{the} feasible optimal \osrdel{solution}\osr{allocation}, stopping the iterations early (e.g., when the residuals fall below the convergence threshold $\gamma$) can leave very small constraint violations in the final rate allocation $x_*$. It is important to resolve these violations before the solution is safe to push out on the network. Therefore, \sysname \osrdel{computes the projection of}\osr{\textit{projects}} the nearly-feasible output into a strictly feasible one \osr{ while
minimizing the loss in objective.} \osrdel{In order to achieve this, GATE trims rate of multiple paths.}}

\rbnew{\sysname\osr{'s projection technique} resolves constraint violations in three steps (\osrdel{as described}\osr{detailed} in \cref{alg:projection}):
\begin{enumerate}
    \item \textbf{Non\osr{-}\osrdel{ }Negativity Constraint\osrdel{s} \osr{(\cref{eq:non-negative-primal})}:} \osrdel{First, we project rates on all paths to [0, $\infty$)}\osr{Since each path's rate must be non-negative, we project all rates to [0, $\infty$)}.
    \item \textbf{Demand Constraint\osrdel{s} \osr{(\cref{eq:demand-constraint-primal})}:} \osr{If any rate allocations exceed requested demand
    we must trim the excess. There is a choice of paths to trim from, and \sysname greedily selects the path which has the most links with their capacity constraints currently violated (denoted as $C_p$), thereby reducing the subsequent capacity projections required.}\osrdel{Next, for each commodity that was allocated more than it requested, we trim the excess. We greedily select paths to trim, in descending order of their scores ($score_p$).} 
    \item \textbf{Capacity Constraints \osr{(\cref{eq:capacity-constraint-primal})}:} After demand is bounded, \osr{we resolve capacity constraints for each edge by trimming flows passing through that edge. We greedily pick flows such to minimize the reduction in the TE objective \srnew{couldn't parse this sentence}. \bgc{I think just delete ``such''?} For max-flow ($\alpha=0$), all flows are equally favorable to be trimmed. For max-min fair ($\alpha\rightarrow\infty$), the flow with the highest demand served ($S_f$) is the best candidate for trimming. Generalizing to any $\alpha$, the flow with the largest $S_f^\alpha$ gets trimmed.} \osrdel{we resolve capacity constraints with a similar strategy. For each edge where the traversing paths exceed capacity, we trim rates on paths (selected in descending order of their scores) until the violation is resolved.}
\end{enumerate}
\osr{We project all flows in parallel to take full advantage of the available GPU parallelism, though we considered a serial variant which enables some small optimality benefits (\cref{s:appendix:extended-design:projection}).}
}

\subsection{Optimality Guarantee}\label{s:design:opt-guarantee}
\osrdel{\sysname's \textit{iterative, local sub-problems} converge to the \textit{globally optimal} max-min fair allocation. We provide the following theorem:}

\begin{theorem}
\label{thm:alpha_admm_conv}
\bg{And now, what do we say about $\alpha=\infty$?} \rb{Fixed} \rbnew{\osrdel{For any $\alpha$-fair objective (including the max-min fair limit as $\alpha\to\infty$), \sysname's iterations asymptotically guarantee the following:\\
A. \emph{Constraint feasibility:} Capacity and demand constraints are satisfied, with the projection scheme resolving violations arising from early stopping under a numerical threshold $\gamma$.\\
B. \emph{Objective convergence:} The objective defined in \eqref{eq:objective-primal} reaches the optimal value (within $\gamma$-tolerance).\\
C. \emph{Consensus constraint:} The rate for a path $x_{r}$ and all rate suggestions $y_{e,r} \forall e \in r$ are equal.}}
\osr{For every topology and demand matrix, GATE computes the global-optimal allocation within a $\gamma$-tolerance.}
\end{theorem}

\bgc{(1) I don't think $\gamma$-tolerance is a standard formally defined thing. I guess you mean there is either additive or multiplicative error. It would be better to write ``GATE's computed allocation is $\geq(1-\gamma)OPT$ where $OPT$ is the optimal...'' or something like that as appropriate. (2) I think we want to convey that we can get arbitrarily close to the optimal. So it would be better to state that explicitly in the theorem, if indeed we proved it that way: ``For any $\gamma>0$, GATE's computed allocation...''}


\osrdel{\textbf{Proof sketch:} Detailed in \cref{appendix:thm proof}.}
\osr{\noindent We detail the proof in \cref{s:appendix:extended-design:optimality}.}

\subsection{Implementation}\label{s:design:implementation}
We implement \sysname in \textasciitilde1,500 lines using PyTorch and CUDA. The implementation expects the inputs described in \cref{alg:gate} -- capacity edgelist, demand matrix, set of pre-computed paths. In our implementation, we use $A=10$ and $C=2$ in the adaptive learning scheme (\cref{eq:adaptive-learning}), taking inspiration from ADMM literature~\cite{neal2011distributed}. \osr{As discussed in \cref{s:evaluation:setup}, we tune $\gamma$ for every topology to optimize for DAO (i.e. the right tradeoff between runtime and optimality).}\osrdel{We select $\gamma$ depending for each topology, to optimize for the right tradeoff between runtime and optimality (DAO).}
\osrdel{We find $\gamma = 10^{-3}$ as a good convergence threshold in our evaluations.} \sysname outputs the rate\osrdel{s over} \osr{allocated to} each path.

%% file: 4-evaluation.tex
\section{Evaluation}\label{s:evaluation}
We aim to evaluate two facets of \sysname's performance: (1) how well \sysname's traffic engineering allocation satisfies the goals and priorities operators have (\cref{sec:motivation:te-objectives}) for a given network, and (2) how well this performance holds up as network sizes grow. We first formally describe the metrics we use to evaluate performance on operator goals. We then compare \sysname to other SOTA TE systems, examine how \sysname converges over iterations, examine its scaling behavior, and conclude by presenting a factor analysis showing how \sysname's design elements work together to enable its effective performance.

\subsection{Metrics}\label{s:evaluation:metrics}

We rely on two standard independent metrics – optimality and runtime – to characterize TE solutions. We additionally introduce a new third metric, Drift Adjusted Optimality, that captures their joint impact.

\noindent \textbf{Optimality.} We report optimality as the distance from the optimal allocation ($OPT$) for the given objective, typically max-min fair. For an allocation $X$, we compute $min(\frac{X_d}{max(OPT_d, \vartheta)}, 1)$ for every demand $d$, and report optimality as the mean across all demands, in-line with metrics used to evaluate production TE systems \cite{b4, swan}. This is a modification of the $q_\vartheta$ metric used in recent works \cite{soroush, marcus2019neo, lu2021pre} to not penalize allocations for serving more traffic than the optimal allocation. \osr{For TE solvers with a max-flow objective, optimality measures total throughput served relative to an optimal max-flow allocation}.

\vspace{0.2em}
\noindent \textbf{Runtime.} We report runtime as the time taken by each algorithm to compute the allocation assuming a set of paths have already been computed (matching prior work~\cite{ncflow,soroush,DeDe, swan, Danna}).  We measure \osrdel{CPU} runtimes \osr{for CPU-based solvers} on a 128-core AMD EPYC with 240GB memory. \osr{For GPU-based TE solvers, we use} \osrdel{We measure GPU runtimes with }a \osr{machine with one} \osrdel{1x }NVIDIA A100 40GB GPU.

\vspace{0.2em}
\ak{alternative opening wording: "Intuitively, this reflects the reality that demands shift over time; thus, while path allocations are computed for a particular state of demand, they end up serving the later demand actually present in the network, which drifts further with time. ..."}
\noindent \textbf{Drift Adjusted Optimality (DAO).} \osr{DAO reports the optimality of the allocation generated by a TE solver, when \textit{applied} to the actual demand and link capacities by the time computation is complete.}
A TE \osrdel{algorithm}\osr{solver} uses the network state \osr{$(D, C)_0$} at time $T=0$ as input\osrdel{ ($demands_{0}$, $capacities_{0}$)}, and \osrdel{generates}\osr{computes} an allocation \osr{$X$ }after runtime $K$. In that \osrdel{time}\osr{period}, network state\osrdel{ can} drift\osr{s} to \osrdel{($demands_{K}$, $capacities_{K}$)}\osr{$(D, C)_K$. Therefore, what actually gets deployed on the network, would be a \textit{trimmed} $X$ that fits within the demand and capacity constraints of the new network state. We perform the trimming in two steps: first, we scale down each node-pairs' paths $P_{st}$ (in equal proportions), so that they do not exceed $D_{st,K}$, \osrdel{We define DAO as a metric to capture the effect of this drift on the actual optimality experienced by flows, computing $DAO_{K}$ as the optimality of an allocation using inputs from $T_{0}$ applied to the conditions at $T_{K}$.}
$$x_r' = x_r \cdot \min\left(\frac{D_{st,K}}{\sum^{P_{st}}_a{x_a}}, 1\right),\hfill r \in P_{st}$$
Second, we scale down flows traversing any link $e$ whose new capacity ($C_{e, K}$) is being exceeded. Since a flow might traverse several such \textit{congested} links, we scale down a flow with the factor of the most \textit{congested} link on its path,
$$x_r^{''} = x_r' \cdot \min_{e \in r}\left( \min\left(\frac{C_{e, K}}{\sum_{p: e \in p} x_p'}, 1\right) \right)$$
Note that \textit{fitting} the allocation to the new network state does not react to increases in capacity or demands, since that requires recomputing the allocation.
Finally, $DAO_K$ \srnew{how exactly? integrating over $K$ or one-shot based on $(D,C)_K$?} is computed as the \textit{Optimality} of this \textit{drift-compatible} allocation $X^{''}$ relative to the optimal allocation for the inputs the solver started with, $OPT_0$. We discuss the choice of $OPT_0$ in \cref{s:appendix:evaluation:dao-definition}. It trivially follows that $DAO_0$ is the optimality of the solver. }\osrdel{This}\osr{DAO} effectively captures the \osr{performance} penalty \osrdel{in optimality} incurred by algorithms \osr{due to their runtimes, and applies it to the solver optimality.}\osrdel{for taking longer to solve.} We note DAO is a conservative metric (i.e. overly kind), as drift increases further after the path allocations get programmed until the next solution is computed, whereas we effectively assume drift stops at programming time. \osr{We further show how the runtime penalty impacts DAO in \cref{s:appendix:evaluation:dao-definition}}.\ak{last sentence unclear what that might mean....}

\subsection{Experiment Setup}\label{s:evaluation:setup}

\begin{figure*}[t]
\centering
\captionsetup{font=footnotesize, labelfont=normalfont, textfont=normalfont} 
  \begin{minipage}[b]{0.31\textwidth}
    \centering
    \renewcommand{\arraystretch}{1.1}
    \rowcolors{1}{gray!30}{gray!10}
    \begin{tabular}{lrr}
        \textbf{Topology} & \textbf{\#Nodes} &
         \textbf{\#Edges} \\ \hline
        Abilene  & 11 & 28 \\
        \geant    & 23 & 122 \\
        ESNet    & 68 & 158 \\
        Cogentco & 197 & 478 \\
        \wansmall & $O(100)$ & $O(100)$ \\
        \wanlarge & $O(1000)$ & $O(1000)$ \\
    \end{tabular}
    \vspace{0.1in}
    \captionof{table}{WAN topologies used for \sysname's evaluation.}
    \label{t:topology-sizes}
  \end{minipage}
  \hfill
  \begin{minipage}[b]{0.31\textwidth}
    \includegraphics[width=\linewidth]{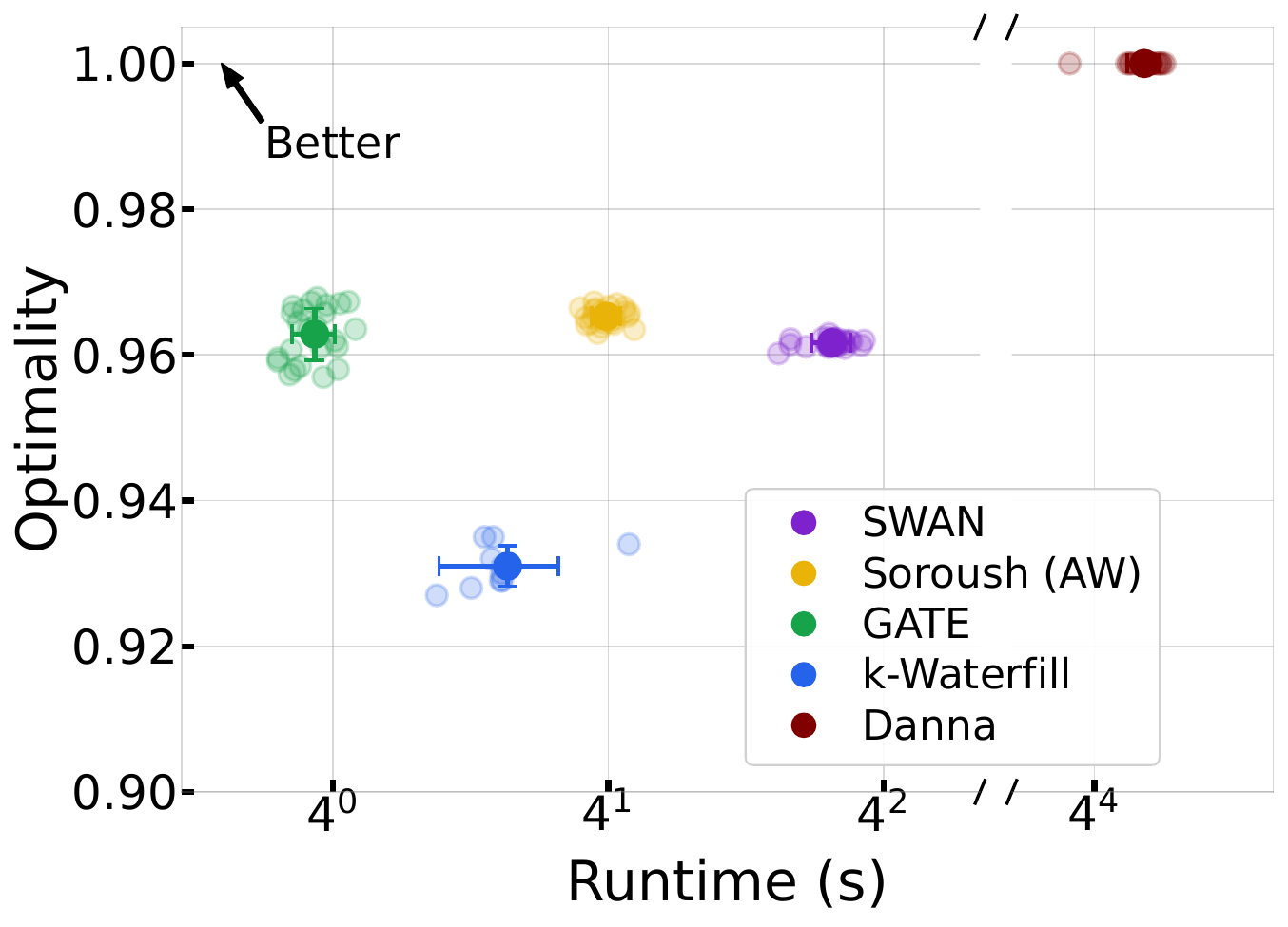}
    \caption{Runtime and Optimality for different TE algorithms over \wanlarge.}
    \label{fig:runtime-v-optimality-b2}
  \end{minipage}
  \hfill
  \begin{minipage}[b]{0.31\textwidth}
    \includegraphics[width=\linewidth]{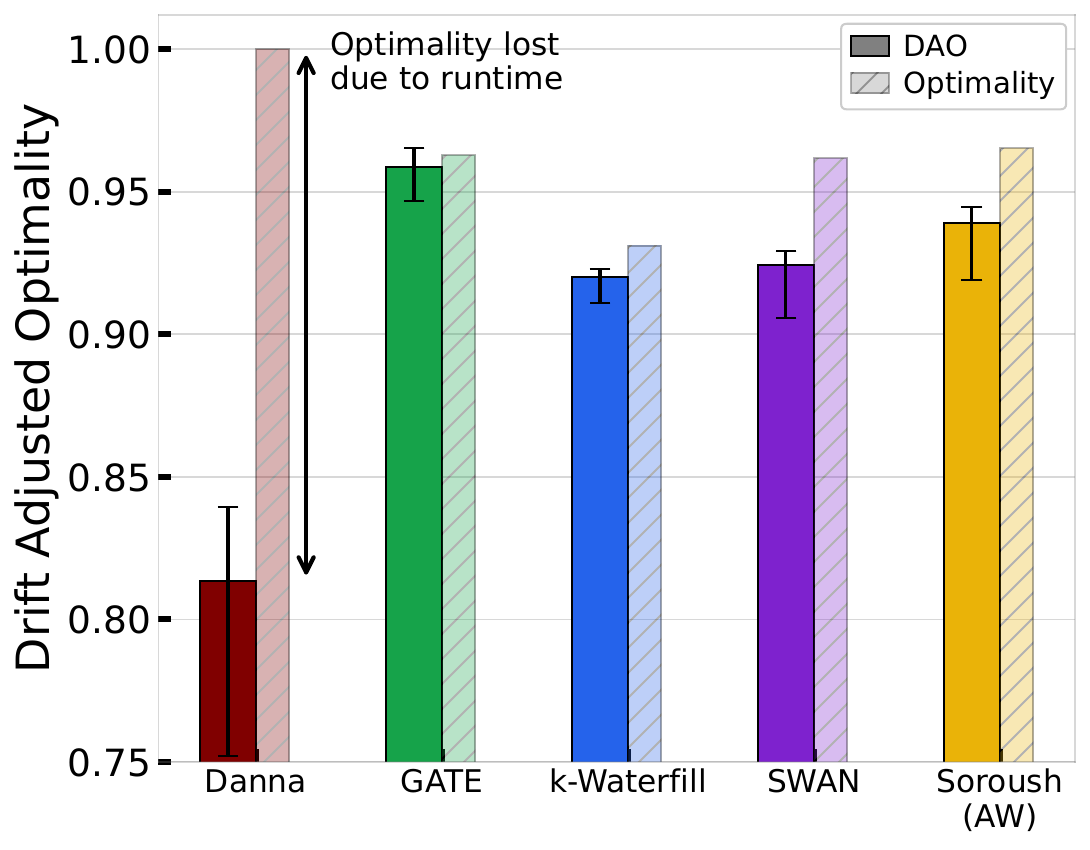}
    \caption{Drift Adjusted Optimality for different TE algorithms over \wanlarge.}
    \label{fig:dao-v-algorithm}
  \end{minipage}
\end{figure*}

\begin{figure*}[t]
\centering
\captionsetup{font=footnotesize, labelfont=normalfont, textfont=normalfont} 
  \begin{minipage}[b]{0.31\textwidth}
    \includegraphics[width=\linewidth]{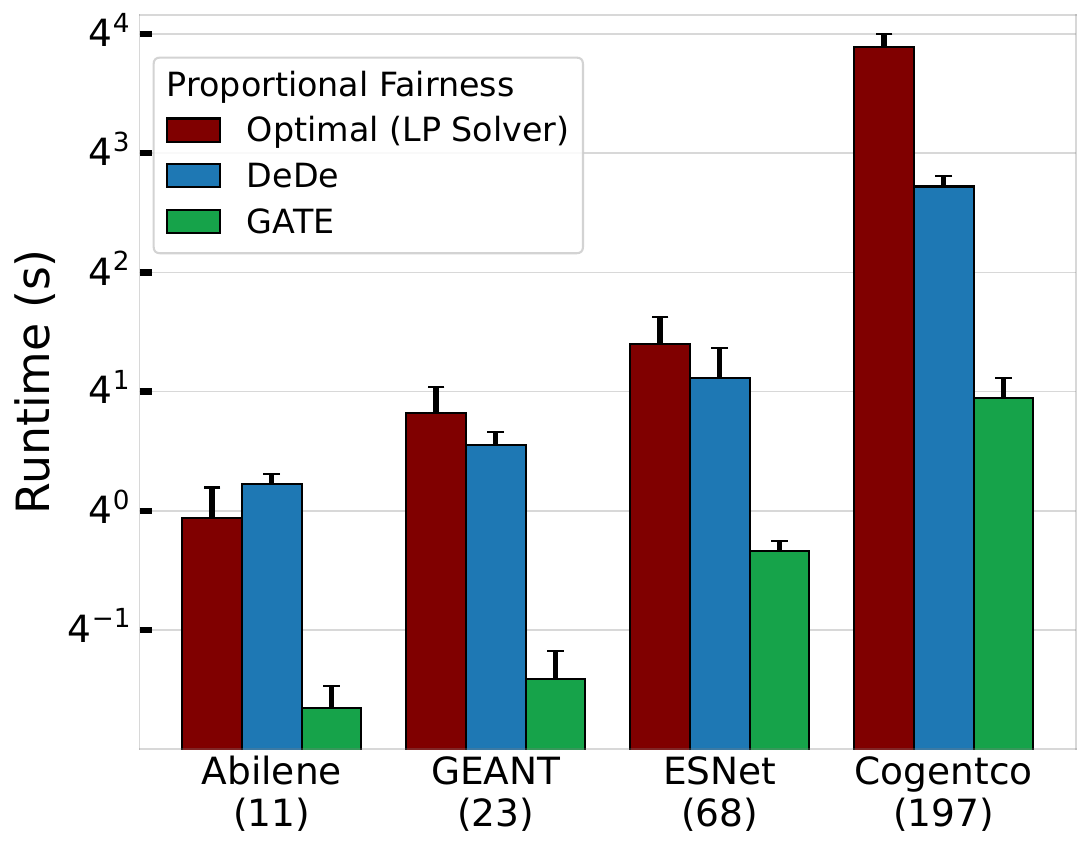}
    \caption{Runtimes for decomposition-based TE algorithms with proportional fairness objective across topologies of various sizes.}
    \label{fig:decomp-style-eval}
  \end{minipage}
  \hfill
  \begin{minipage}[b]{0.31\textwidth}
    \includegraphics[width=\linewidth]{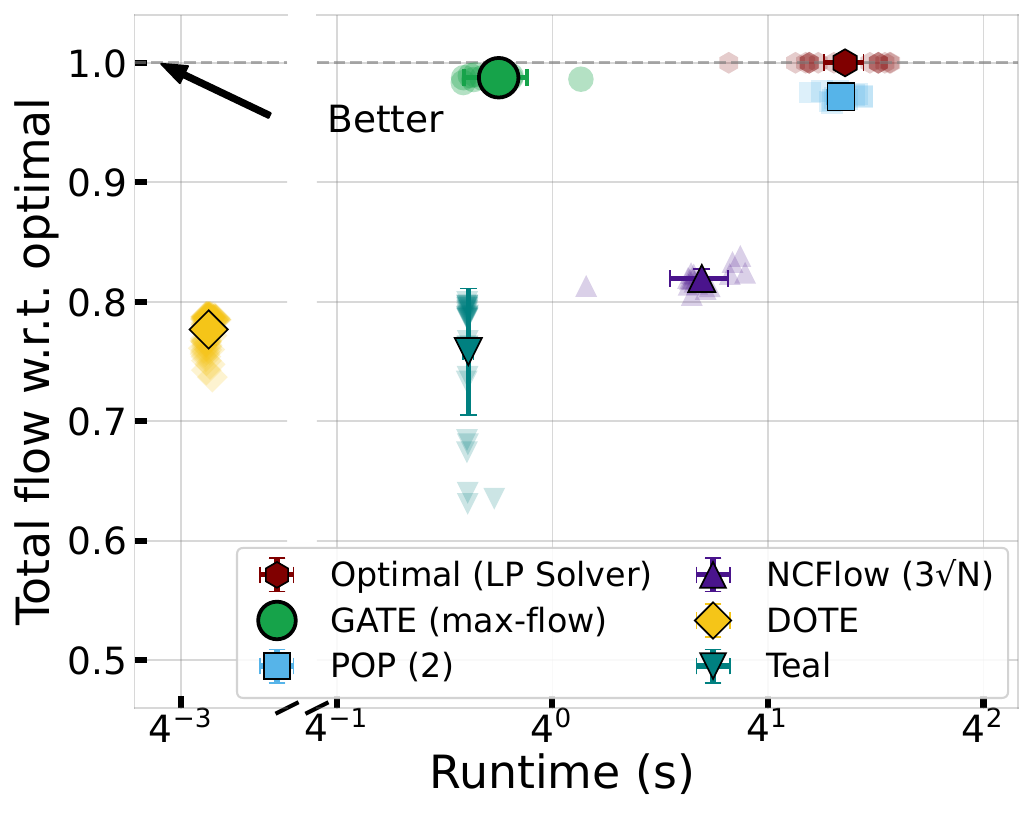}
    \caption{\osr{Runtime and Optimality for max-flow TE solvers over \wanlarge with \emph{16x scaled traffic} (large network, high load).}}
    \label{fig:runtime-v-throughput-b2-16x}
  \end{minipage}
  \hfill
  \begin{minipage}[b]{0.31\textwidth}
    \includegraphics[width=\linewidth]{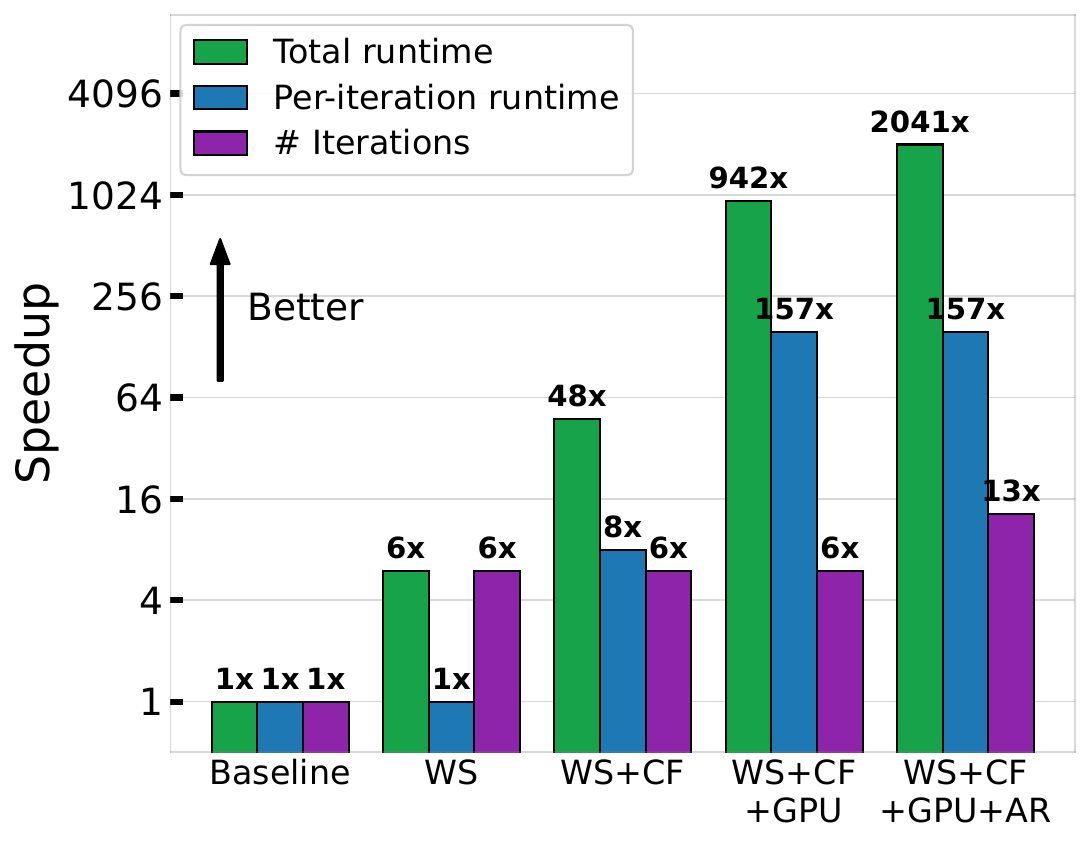}
    \caption{\sysname's speedup with Warm Start (WS), Closed Form (CF), GPU, and Adaptive Rate (AR). Note log scale y axis.}
    \label{fig:tuning-benefits}
  \end{minipage}
\end{figure*}
We evaluate \sysname over a wide range of network topologies and sizes. We use data from four publicly-available WAN topology datasets, as well as data sourced from two production topologies (\wansmall and \wanlarge) operated by a large-scale network operator, with sizes shown in~\cref{t:topology-sizes}.

\noindent \textbf{TE Inputs.} For \wansmall, we collected the live production stream of aggregated topology and demand data that was taken as input by the operator's production TE system over three months in 2025, with data points every twelve seconds. Likewise for Abilene, we use the publicly available topology and demand snapshots~\cite{sndlib, topohub}, reported every five minutes.

For the rest of our datasets, the available traffic demands are less complete. Specifically, for EsNet, Cogentco, and \geant, we use the static topology provided and estimate a realistic demand matrix via the gravity model used in prior work~\cite{gravity-model, ncflow}. For \wanlarge, we pull historical snapshots of representative demand and topology every three months, over a three year period from January 2020 to December 2022 to capture network growth. Then, to produce a series of realistically-spread timestamped matrices and topologies, we measure the distribution of demand flow shifts (as percentages) between subsequent datapoints observed in production for \wansmall, and sample from that distribution. The result is a series of 1000 realistic input snapshots.
\rbnew{For all topologies, we use the k-shortest paths as input, with k=4 restricted by hardware constraints.}
We use these as the inputs to \sysname and the SOTA algorithms as comparison.


\vspace{0.2em}
\noindent \textbf{Streaming data.} 
To evaluate \sysname's operation in a continuous telemetry streaming regime, we approximate the underlying telemetry datastream for each of our datasets by interpolating the difference between each pair of TE input data points: for every consecutive pair of snapshots separated by time $T$ with $N$ values updated, we uniformly spread the $N$ updates in time $T$.

\vspace{0.2em}
\noindent \textbf{Benchmarks.}
We compare \sysname with recent SOTA \osr{max-min fair }TE \osrdel{solutions}\osr{solvers}: Soroush~\cite{soroush}, SWAN~\cite{swan}, k-Waterfill~\cite{k-waterfill}, and Danna~\cite{Danna}. Next, we compare \osr{the max-flow configuration of} \sysname with \osr{SOTA max-flow solvers:} \osrdel{DeDe}\osr{DOTE\cite{perry2023dote}, Teal\cite{teal}}, NCFlow~\cite{ncflow}, and POP~\cite{pop}. \osr{Finally, we compare \sysname with DeDe~\cite{DeDe}, PDMCF~\cite{zhang2025solving}, and Lin~\cite{lin2006utility}} to evaluate performance against other decomposition techniques. \osr{For all SOTA TE solvers, we use their open source implementations}\osrdel{For all the TE solutions, we used the open-source implementations by} \cite{soroush, DeDe, zhang2025solving}. To be fair while comparing, \osr{we tune each TE solver to a configuration that achieves the highest DAO, discussed in \cref{s:appendix:evaluation:tuning}.}\osrdel{To be as fair as possible in comparisons, for each topology we perform parameter tuning for}\osrdel{and pick the parameters that result in the highest DAO. We use $k=1$ for k-Waterfill, which has the lowest runtimes.} We warm-start SWAN, Danna, \sysname, and DeDe with the optimal allocation from the network state 300 seconds prior.\footnote{We note that in practice, \osrdel{the warm-start can be done}\osr{TE solvers can be warm-started} from a more recent allocation (depending on the algorithm runtime), but we choose 300 seconds as a conservative measure.}

\subsection{Cumulative Performance}\label{s:evaluation:runtime-v-optimality}

We first look at \sysname's cumulative performance on the series of aggregated inputs a TE system would receive for each of our dataset networks. 
    
\vspace{0.2em}
\noindent\textbf{\sysname pushes the Pareto frontier of the runtime optimality tradeoff compared to existing TE solutions.} \cref{fig:runtime-v-optimality-b2} compares the mean runtime and optimality of \sysname and other TE solutions over production traces from \wanlarge. \sysname offers a new operating point, considerably stretching the Pareto frontier closer to optimal. The error bars represent the $p5$ and $p95$, and \osrdel{the} shaded points represent individual runs (randomly sampled to avoid clutter). \sysname takes \textasciitilde \osrdel{1sec}\osr{0.9s} to compute an allocation that is $95\%$ of optimal, achieving a speedup of \osrdel{5}\osr{4}-10$\times$ compared to SWAN and Soroush to achieve similar quality solutions. Danna achieves the optimal allocation, but incurs \osrdel{$\approx 200\times$}\osr{$\sim400\times$} longer runtimes compared to \sysname.  K-waterfilling \osr{is both slower ($2\times$) and 3-5\% suboptimal than \sysname}\osrdel{has similar runtimes as \sysname, but gets 3-5\% worse optimality than \sysname}, \osrdel{with}\osr{while} serving 8\% less demand\osrdel{ than \sysname}, which translates to \$\osr{O(100)} millions in increased network provisioning costs. \cref{fig:runtime-v-optimality-cogentco} shows \osrdel{these results hold and even improve further on Cogentco, a smaller network with uniform capacity links.}\osr{\sysname's optimality improves while operating on networks with uniform capacity links. While not true in practice, it is an important property to be explored in future work.}

\osr{\textbf{Note:} Recall from \cref{s:evaluation:setup} that we tune TE solvers to the configuration that achieves the highest DAO. Therefore, the baselines may appear at a different operating point than reported in prior works. Additionally, we report results at production traffic instead of scaled traffic, unless specified otherwise. For consistency with results from prior works, we report performance of a few parameter settings that we evaluated on varying traffic loads in \cref{s:appendix:evaluation:result-consistency}.}



\vspace{0.2em}
\noindent \textbf{\sysname achieves the highest DAO amongst evaluated \osr{max-min fair }TE \osrdel{solutions}\osr{solvers}.}
\cref{fig:dao-v-algorithm} shows DAO of different algorithms computed from the same experiment as \cref{fig:runtime-v-optimality-b2} (over \wanlarge), showing the combined impact of runtime and optimality. The shaded values represent optimality of the algorithm, to indicate the degradation in optimality due to runtime. Danna experiences a 14\% degradation from \osrdel{its} optimal \osrdel{value }because of \osrdel{the}\osr{its} significantly long runtimes ($O(100)$ seconds). SWAN and Soroush experience a noticeable degradation. While their optimality was comparable to \sysname, their DAO is 3-4\% lower. \sysname and k-Waterfill experience the least degradation because of their low runtimes\osrdel{. However,}\osr{, but} k-Waterfill \osrdel{had}\osr{has inherently} lower optimality\osrdel{ to begin with}. \osr{We discuss this in more detail in \cref{s:appendix:evaluation:dao-cdf}.}
\osrdel{\sysname's shorter runtime enables it to perform disproportionately better in the worst scenarios, which is visible in the lower percentiles of the CDF of DAO in \cref{fig:dao-cdf}, corresponding to cases of large capacity changes. This is because we both reconverge quickly, and to a better state. Note that each data point in this CDF is a single snapshot; this actually dilutes some of \sysname's fast reaction benefits, because a link failure only affects some flows in a snapshot. Fast reaction to failure is more directly captured by runtime. We provide a comparison of total traffic served across algorithms in \cref{sec:appendix:eval-continued}.}



\vspace{0.2em}
\noindent \textbf{\sysname is an order of magnitude faster than the prior decomposition based approaches.} \cref{fig:decomp-style-eval} compares the runtime of \sysname to the SOTA decomposition-based TE, DeDe, as well as an exact solver, across different topologies\footnote{The DeDe implementation failed to run (crashed) on \wansmall and \wanlarge, though we do not believe this is a fundamental problem.}. We compare on proportional fairness rather than the more common max-min objective because DeDe cannot solve for max-min. 
Both DeDe and \sysname get within 1.5\% optimality of the exact solution. However, \sysname is consistently $10\times$ faster than DeDe across a range of different topologies, in large part due to \sysname (1) not using Gurobi, as our subproblem formulation is further decomposed, \bgnew{(2) leveraging GPU,} and (3) using adaptive learning rates, as we discuss in \cref{s:evaluation:factoranalysis}. \osr{For completion, we further report a comparison of \sysname with two more decomposition-based routing solvers\cite{lin2006utility, zhang2025solving} (TE-like) in \cref{s:appendix:evaluation:decomposition-te}.}
\osrdel{We compare PDMCF\mbox{\cite{zhang2025solving}} and Lin\mbox{\cite{lin2006utility}} for completion, despite their limitations described in \mbox{\cref{sec:motivation-existing}}.  Lin\mbox{\cite{lin2006utility}} converges to a severely suboptimal ($\geq 20\%$) allocation, with $6x$ longer runtimes than \sysname. Its performance substantially degrades on heterogenous networks (\wansmall and \wanlarge). PDMCF\mbox{\cite{zhang2025solving}} serves more throughput than \sysname on all but \wansmall and \wanlarge, since it expects infinite demand and infinite paths in its formulation. However, it takes nearly $10x$ longer to converge on the same threshold $\gamma=10^{-3}$. On heterogenous networks (\wansmall and \wanlarge), PDMCF stagnates at a suboptimal value much lower than \sysname, despite the relaxed constraints. Note that we equip DeDe, PDMCF, and Lin with \sysname's projection scheme to fairly compare the decomposition technique.}

\vspace{0.2em}
\noindent \textbf{\osr{On large networks, \sysname achieves the highest DAO amongst evaluated max-flow TE solvers.}}
\osr{\sysname supports a spectrum of $\alpha$-fair utilities (\cref{eq:objective-primal}), including max-flow ($\alpha=0$). While \sysname's runtime benefits are more prominent at fairer objectives (high $\alpha$), we compare \sysname with max-flow TE solvers for completeness. \cref{fig:runtime-v-throughput-b2-16x} compares the runtime and total throughput relative to optimal solver, served by \sysname and other SOTA max-flow TE solvers over \wanlarge with $16\times$ scaled demand. \sysname achieves the highest DAO by serving \textasciitilde$98.9\%$ of the optimal throughput in $<700$ms. The next most optimal approximation, POP, takes $5\times$ longer runtime and serves $1.5\%$ less throughput compared to \sysname. \ak{the rest of the text from here onwards reads awkwardly. what is POP-4/8/16. we say "note that" too many times here, you can only say that once.}\rb{Fixed}Other configurations of POP (with more partitioning) have faster runtimes but serve much less throughput, resulting in a worse DAO. Learning-based TE solvers (Teal, DOTE) have significantly lower runtimes: Teal takes 580ms, and DOTE takes 24ms. However, their throughput served is much lower compared to other TE solvers, $77\%$ and $78\%$ of optimal respectively. We make an observation that max-flow is an \textit{easier} problem to solve compared to max-min-fair. Since \sysname's performance benefits start to reflect better at scale, \cref{fig:runtime-v-throughput-b2-16x} intentionally reports the performance on a \textit{harder} configuration, i.e. a large network under high load. We report more configurations (smaller network, production loads) in \cref{s:appendix:evaluation:max-flow}. On small networks, exact solvers are fairly quick (\textasciitilde$35ms$),
while \emph{\sysname can take relatively longer}. This is because \sysname incurs overheads with each iteration that amortize with network size. This is consistent with the comparison of max-min fair TE solvers on small networks (\cref{fig:runtime-v-topology}).}
\subsection{Looking Inside \sysname's Convergence}\label{s:evaluation:iterates-behavior}
\begin{figure}[]
    \centering
    \includegraphics[width=0.98\linewidth]{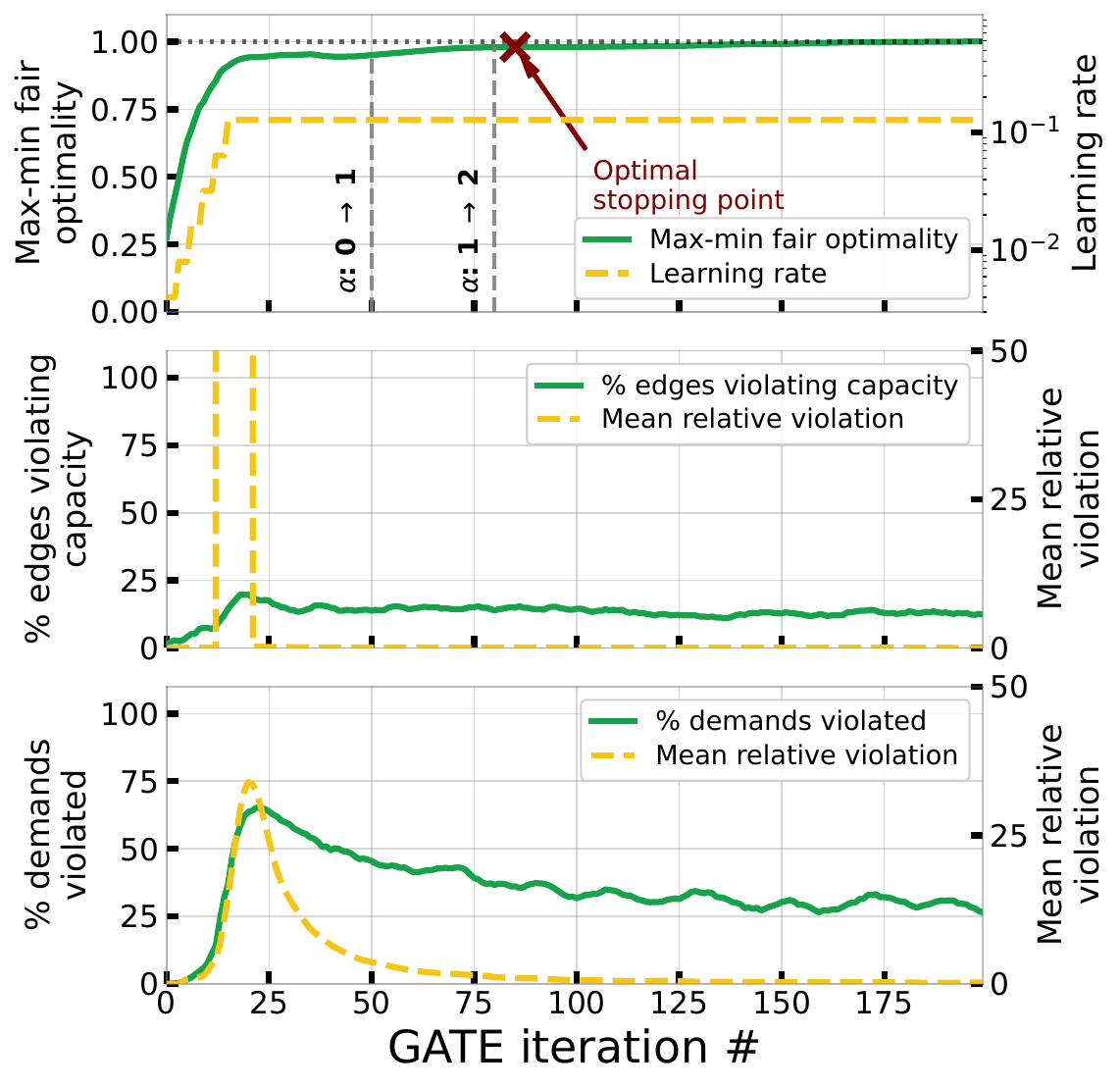}
    \caption{\sysname's intermediate values with \wansmall \osrdel{as input} -- (a) top: \bgnew{optimality after projection} and learning rates, (b) middle: Capacity violations (without projection) and, (c) bottom: Demand violations (without projection). \bg{The subfigures are missing labels (a),(b),(c)} \bg{If possible, add a vertical line in the figure to show that we did increment $\alpha$ to 3, just at the time of stopping.}}
    \label{fig:gate-internals}
\end{figure}

\begin{figure*}[t]
\centering
\captionsetup{font=footnotesize, labelfont=normalfont, textfont=normalfont} 
  \begin{minipage}[b]{0.33\textwidth}
    \includegraphics[width=\linewidth]{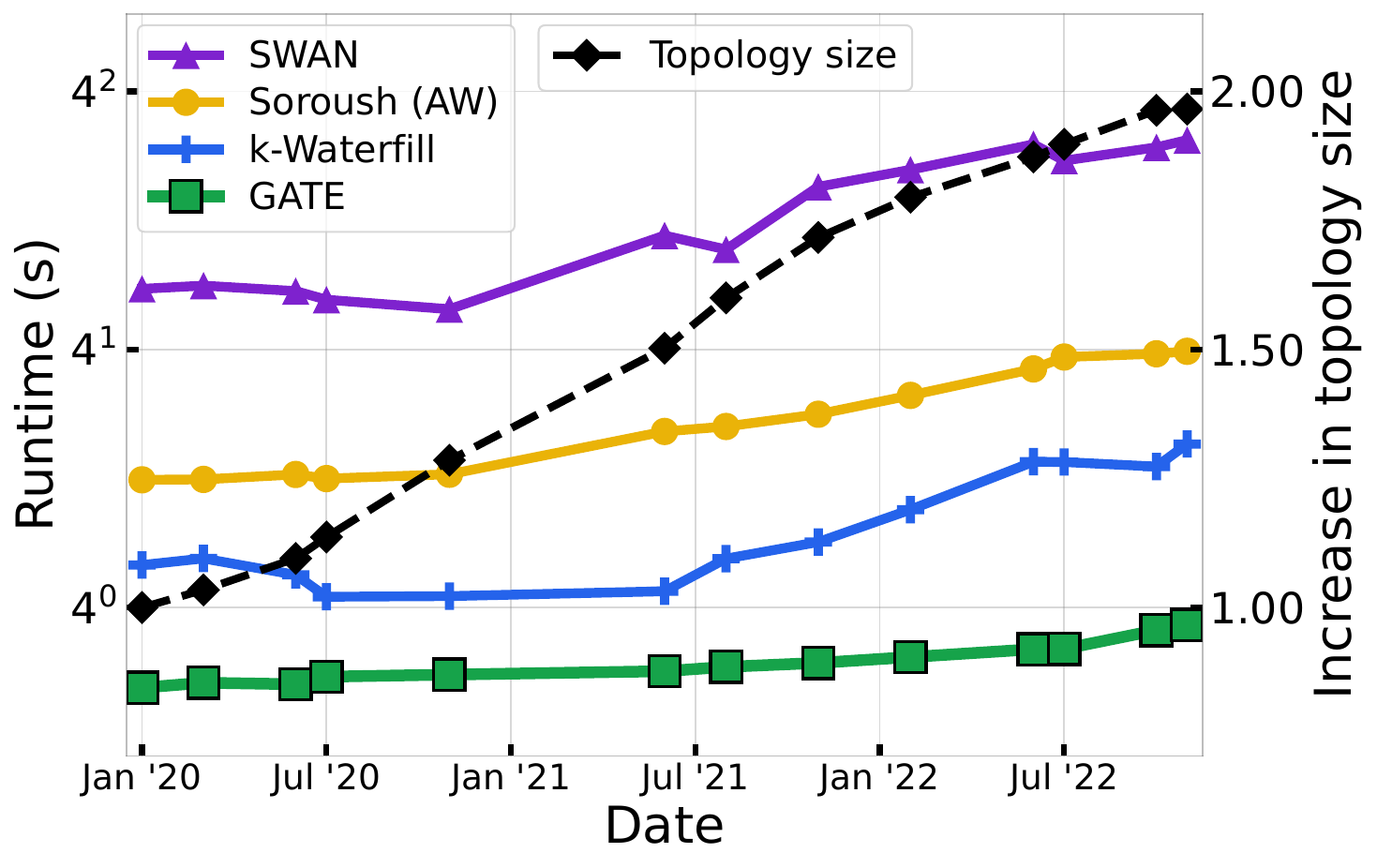}
    \caption{Runtimes of TE \osr{solvers}\osrdel{algorithms} over a \osr{3 year }period \osrdel{of 3 years} as \wanlarge grows in number of routers.}
    \label{fig:runtime-v-topology-size-linear}
  \end{minipage}
  \hfill
    \begin{minipage}[b]{0.30\textwidth}
    \includegraphics[width=\linewidth]{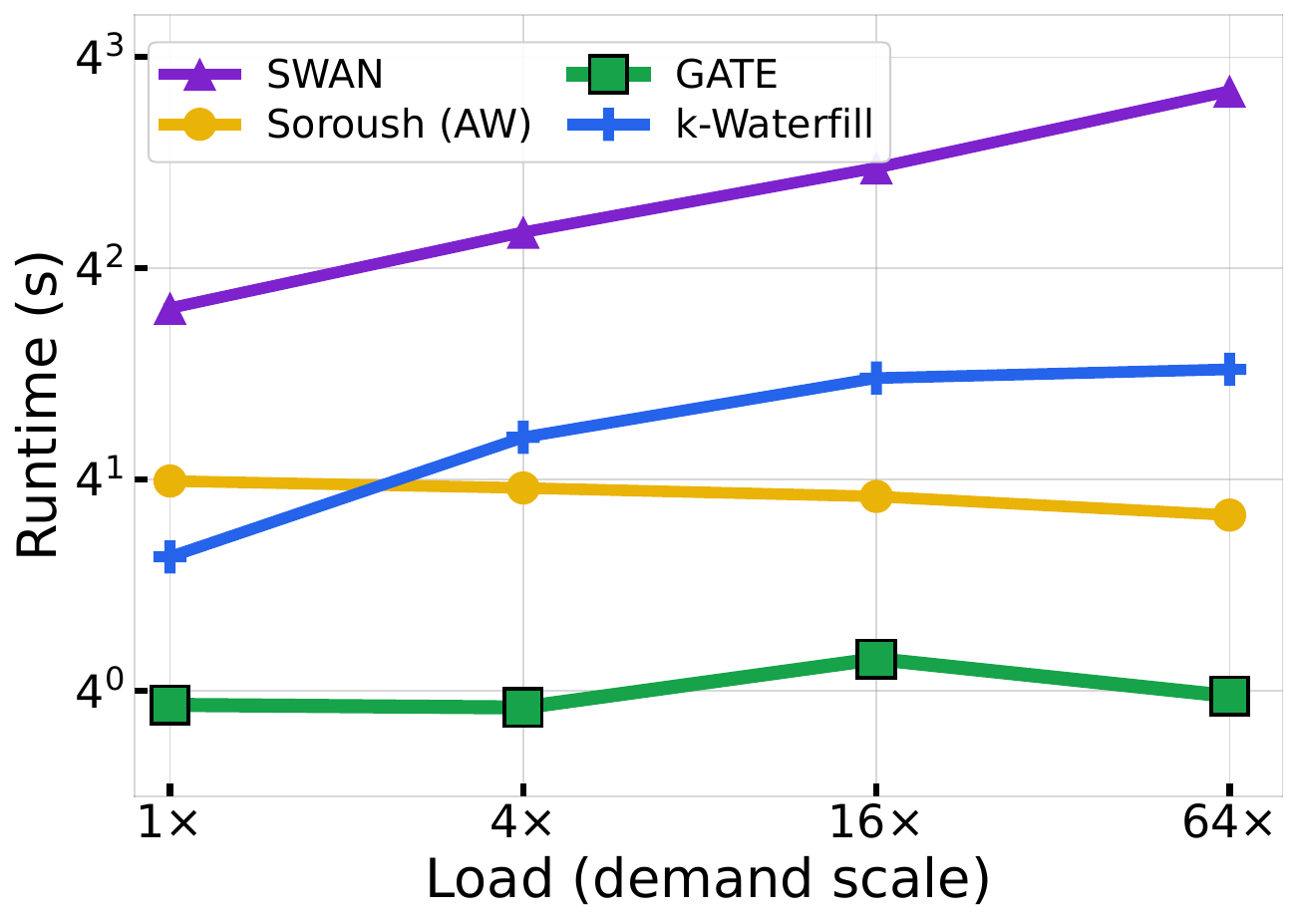}
    \caption{\osr{Runtimes over \wanlarge with increasing network load \textit{(Danna discussed in text).}}}
    \label{fig:runtime-v-load-B2}
  \end{minipage}
  \hfill
   \begin{minipage}[b]{0.30\textwidth}
    \includegraphics[width=\linewidth]{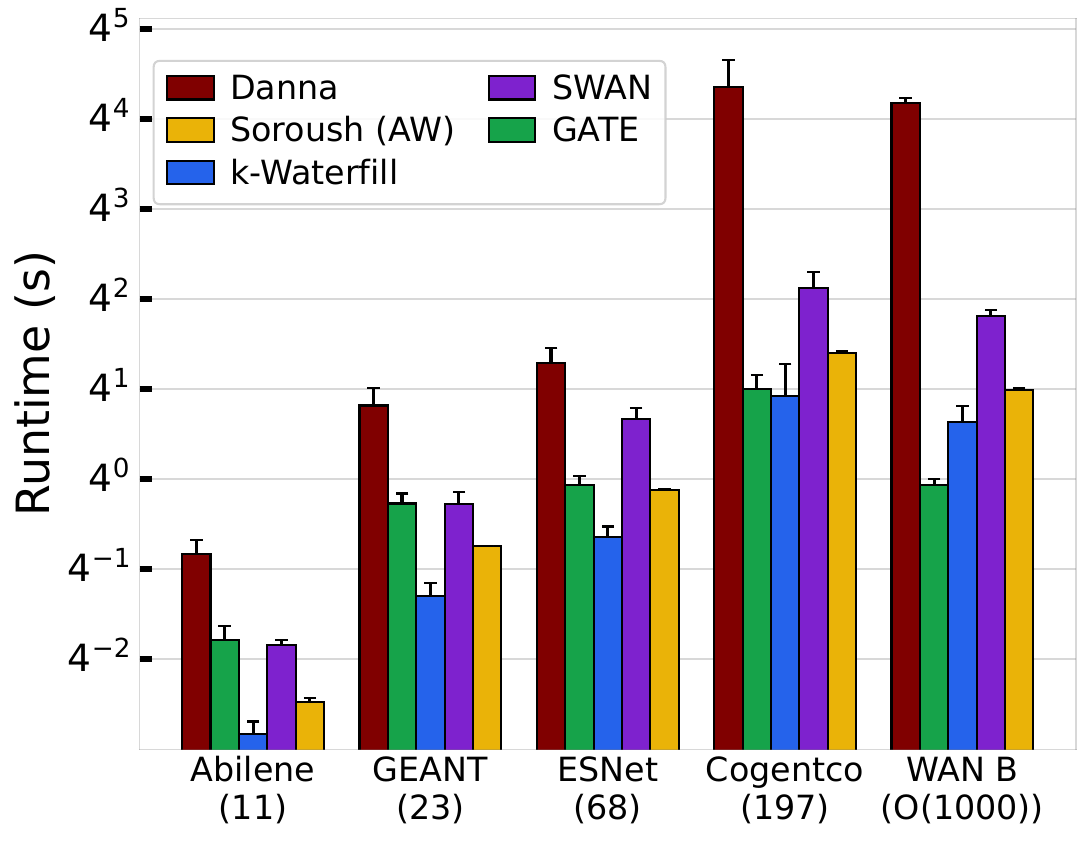}
    \caption{Runtimes for different \osr{max-min fair }TE \osrdel{algorithms}\osr{solvers} across topologies of various sizes.}
    \label{fig:runtime-v-topology}
  \end{minipage}
  \hfill
  
\end{figure*}

\bgnew{We now provide insight into how \sysname works by seeing inside its convergence process, shown in \cref{fig:gate-internals} for \wansmall. Note that \sysname's convergence can be stopped at any point -- either earlier or later than its default stopping criterion (\S\ref{sec:optimizing-convergence}) -- after which projection (\S\ref{s:design:projection}) runs and produces a feasible solution free of constraint violations. The figure shows \emph{post-projection} optimality\footnote{To be clear, ordinarily \sysname only runs projection once at the end. Here, we run projection after each iteration to illustrate what would have resulted from stopping \sysname's convergence after that iteration.} (top) and \emph{pre-projection} constraint violations (middle and bottom).} \osr{For brevity, we discuss more details in \cref{s:appendix:evaluation:gate-iters}.}   

\subsection{Scaling Behavior}\label{s:evaluation:scaling-behavior}
Next, we look at how \sysname's performance varies across different topologies and increasing network sizes.

\vspace{0.2em}
\noindent \textbf{\sysname's runtime grows much more slowly with network size \osr{and demand load }compared to existing SOTA TE approaches.}  \cref{fig:runtime-v-topology-size-linear} shows runtimes for SWAN, Soroush, and \sysname on the historical data from \wanlarge, totaling growth in number of routers by a factor of two over a three-year period. The runtimes for SWAN, Soroush, and \sysname increase by 2.22x, 2.48x, and 1.64x respectively. \sysname enjoys a significantly more modest increase due to increasing parallelization, which keeps the per-iteration runtimes low. As the topology size becomes twice the original, the per-iteration runtime grows only to \textasciitilde1.11x. In that same period, the number of iterations required to converge increases by \textasciitilde1.45x of the original.
\osr{\cref{fig:runtime-v-load-B2} shows a matching trend for a fixed network size as we scale demand by up to 64$\times$; \sysname's fastest overall runtime at base demand stays low across the full scaling range, whereas SWAN and k-waterfill increase by 4$\times$ over the same period. Danna's runtime increased by 20$\times$ as demand scaled to 4$\times$. At 16$\times$ and higher demand scaling, Danna did not finish computation for an hour.}


\vspace{0.2em}
\noindent \textbf{\sysname outperforms other SOTA TE on all but the smallest (sub-30 router) topologies, with a widening gap as network size grows.}
\cref{fig:runtime-v-topology} compares the runtime of different TE over a variety of topologies of different sizes. We observe that while \sysname (0.102s) is slower on smaller topologies like Abilene than SWAN (0.0777s) and Soroush (\osrdel{0.0693s}\osr{0.0349s}) -- as there is less benefit to \sysname's key parallelization technique -- it beats SOTA alternatives on topologies larger than \geant (23 nodes). \aknew{For large topologies such as those used by the hyperscaler in \wanlarge (and therefore representing today’s real-world deployments),} \sysname actually comes to match the k-Waterfill approximation's runtime, while providing \textit{significantly} more optimal traffic placement (\cref{s:evaluation:runtime-v-optimality}).\ak{is last sentence still true? or we beat waterfill runtime?}


\vspace{0.2em}

\subsection{Factor Analysis}\label{s:evaluation:factoranalysis}

We now analyze how \sysname's design decisions contribute to its strong performance. Naively applying ADMM to TE results in exceedingly long convergence times and performance highly sensitive to careful parameter tuning~\cite{neal2011distributed}. We discuss the impact of the techniques we exploit to help.

\vspace{0.2em}
\noindent \textbf{Our optimization techniques provide a cumulative \linebreak $\boldsymbol{2041\times}$ speedup over naive ADMM decomposition.} \cref{fig:tuning-benefits} shows the incremental impact of each additional technique we applied in our design when run on \wanlarge. The total runtime of \sysname $=$ \# iterations $\times$ per-iteration runtime. We observe how each technique independently affects either component of total runtime compared to a baseline without any techniques:
\begin{itemize}[nosep,leftmargin=1.5em]
    \item \textbf{Baseline:} We replaced our closed-form decomposition of Eq. (12-15) with a maximization problem similar to that used by DeDe, and call Gurobi to solve it. Otherwise, everything is identical to \cref{alg:gate}.
    \item \textbf{Warm start (WS)} decreases \# iterations by $6\times$, by seeding with allocations from prior run's output.
    \item \textbf{Closed form (CF)}: We restore our closed-form decomposition, so that Gurobi is no longer needed, but execute the solution in the machine's 64 CPU threads (not GPU). This decreases the per-iteration runtime by $8\times$.
    \item \textbf{GPU}: decreases per-iteration runtime by $20\times$, by enabling massive parallelization from running threads on GPU instead of CPU.
    \item \textbf{Adaptive rate (AR)}: decreases \# iterations by $2.2\times$, by dynamically adjusting learning rate during convergence (i.e., WS+CF+GPU+AR restores the full implementation of \sysname.)
\end{itemize}
\noindent These optimizations combine to provide the cumulative convergence speedup of three orders of magnitude.

%% file: 6-conclusion.tex
\section{Conclusion}
We presented \textbf{\sysname}, a new TE design \bgnew{that addresses the fundamental tension between runtime scalability and solution optimality. \sysname uses a \osrdel{Lagrangian}decomposition \osr{technique}  
augmented with adaptive learning rates and tailored for GPU implementation. GATE's design guarantees convergence to the global optimum and outperforms approximate solvers in both optimality and runtime.}

We believe \sysname provides a valuable new point in the design space that is flexible and future-proof, \bgnew{making it practical for real-world deployment.} 
This is not only because of \sysname's strong performance properties, but also its qualitative benefits: the ability to flexibly navigate runtime {\it vs.} optimality, and to support diverse objectives.   

\bgnew{We see several avenues for future work. Further algorithmic improvements may be possible: {\it e.g.,}
we saw that \sysname converges better on uniform capacity networks than nonuniform (cf. \cref{fig:runtime-v-optimality-cogentco} vs. \cref{fig:runtime-v-optimality-b2}) hence it may be possible to accelerate convergence  with learning rates that adapt based on topological characteristics.}
Future work could also build on \sysname to develop a new paradigm of streaming TE solutions where allocations are continuously refined. Another direction of interest is implementing \sysname distributed across routers. While this affects convergence time due to network communication time between subproblems, it could enable novel distributed TE schemes that provide both optimality and survivability. 

Our work raises no ethical concerns.

%% file: 7-appendix.tex
\newpage
\appendix

\section{Extended Design}
\label{s:appendix:extended-design}

In this section, we provide a more detailed description of the abridged components of \cref{s:design}, including \cref{alg:projection}.

\subsection{Augmented Lagrangian}
\label{s:appendix:extended-design:lagrangian}

\osr{\cref{s:design:problem} introduces the split formulation and states its Augmented Lagrangian~\cite{auglag} schematically, as the primal objective plus one quadratic penalty per constraint family. We give it here in full. The four families are the demand constraint (\cref{eq:demand-constraint-primal}), the copy-based capacity constraint (\cref{eq:capacity-constraint-reformulated}), non-negativity (\cref{eq:non-negative-primal}), and the consensus constraint tying each link's private copies back to the true rate (\cref{eq:consensus-reformulated-true}).}

\osr{
\begin{equation}
\begin{aligned}
    & \mathcal{L}_\beta(x_*,y_*;dual_{*}) = - \underbrace{U^{*}(x_*)}_{\text{primal objective}} \\
    &\quad + \frac{\beta}{2}\sum_{(s,t)} \Big(\underbrace{dual_{dem_{st}} + (\textstyle\sum_{r\in P_{st}} x_r - D_{st})}_{\text{demand constraint residual}}\Big)^2\\
    &\quad + \frac{\beta}{2}\sum_{e\in E} \Big(\underbrace{dual_{cap_e} + \textstyle(\sum_{r:\,e\in r}y_{e,r} - C_e)}_{\text{capacity constraint residual}}\Big)^2\\
    &\quad + \frac{\beta}{2}\sum_{r\in P} \Big(\underbrace{dual_{neg_r} - x_r}_{\text{non-negativity constraint residual}}\Big)^2\\
    &\quad + \frac{\beta}{2}\sum_{e\in E}\sum_{r:\,e\in r} (\underbrace{dual_{con_{e,r}} + (x_r - y_{e,r})}_{\text{consensus constraint residual}})^2
\end{aligned}
\label{eq:augmented-lagrangian}
\end{equation}
}

\subsection{Dual Updates}
\osr{\paragraph{\textbf{Dual updates}}
\cref{s:design:iterative} shows the capacity dual (\cref{eq:dual-updates-cap}) as the representative of this step, and notes that all four families share the form $dual \leftarrow \left[dual + \text{residual}\right]_+$. The remaining three follow, one per constraint family of \cref{eq:augmented-lagrangian}. Each is a scalar accumulation of its own residual with no dependence on any other instance, which is what allows all $|D| + |E| + |P|\!\cdot\!|E| + |P|$ of them to be computed in a single parallel dispatch.}
 
\begin{equation}
\label{eq:dual-updates-dem}
dual_{dem_{st}}^{k+1}= \left[dual_{dem_{st}}^{k}+\Big(\sum_{r\in P_{st}}x_{r}^{k} - D_{st}\Big) \right]_+
\end{equation}
\vspace{-1em}
\begin{equation}
\label{eq:dual-updates-con}
dual_{con_{e,r}}^{k+1}=\left[dual_{con_{e,r}}^{k}+\big(x_{r}^{k}-y_{e,r}^{k}\big)\right]_+
\end{equation}
\begin{equation}
\label{eq:dual-updates-neg}
dual_{neg_{r}}^{k+1}= \left[dual_{neg_{e}}^{k}-x_{r}^{k}\right]_+
\end{equation}
 
\osr{\noindent The demand dual (\cref{eq:dual-updates-dem}) grows while a commodity is allocated more than it requested. The consensus dual (\cref{eq:dual-updates-con}) grows while a path's true rate and a link's copy of it disagree, and is the mechanism by which the split of \cref{s:design:problem} is eventually undone. The non-negativity dual (\cref{eq:dual-updates-neg}) grows while a rate is driven below zero, which is what allows the rate update of \cref{eq:rate-computation} to remain an unconstrained expression.}

\paragraph{\textbf{Slack variables}}
\label{s:appendix:extended-design:slack}
\osr{While applying ADMM, we must convert capacity and demand constraints (\cref{eq:capacity-constraint-primal} and \cref{eq:demand-constraint-primal}) to equality constraints.} \osrdel{Standard ADMM templates rely on \textit{exact block minimizations}, which are harder if constrained by inequalities during the subproblem updates.} We resolve this by introducing non-negative slack variables $slack_{dem} \ge 0$ and $slack_{cap} \ge 0$. In each iteration, these variables are computed to absorb the current inequality gaps such that $\sum x + slack_{cap, e} = C_e$ and $\sum x + slack_{dem, st} = D_{st}$. In each iteration, we compute the slack variables alongside the dual variables ($dual_*$)  by setting $\frac{\delta \mathcal{L}}{\delta slack} = 0$,

\begin{equation}
    slack_{dem_{st}}^{k+1} = \left[ \frac{dual_{dem_{st}}^k}{\beta} + \left( D_{st} - \sum_{r \in P_{st}} x_r^k \right) \right]_+
\end{equation}

\begin{equation}
    slack_{cap_e}^{k+1} = \left[ \frac{dual_{cap_e}^k}{\beta} + \left( C_e - \sum_{r: e \in r} y_{e,r}^k \right) \right]_+
\end{equation}

By mathematically absorbing the overages or underages, the slack variables allow the primal rate updates for $x_r$ and $y_{e,r}$ to be solved as unconstrained minimizations. 

\begin{figure}[ht]
    \centering
    \includegraphics[width=\linewidth]{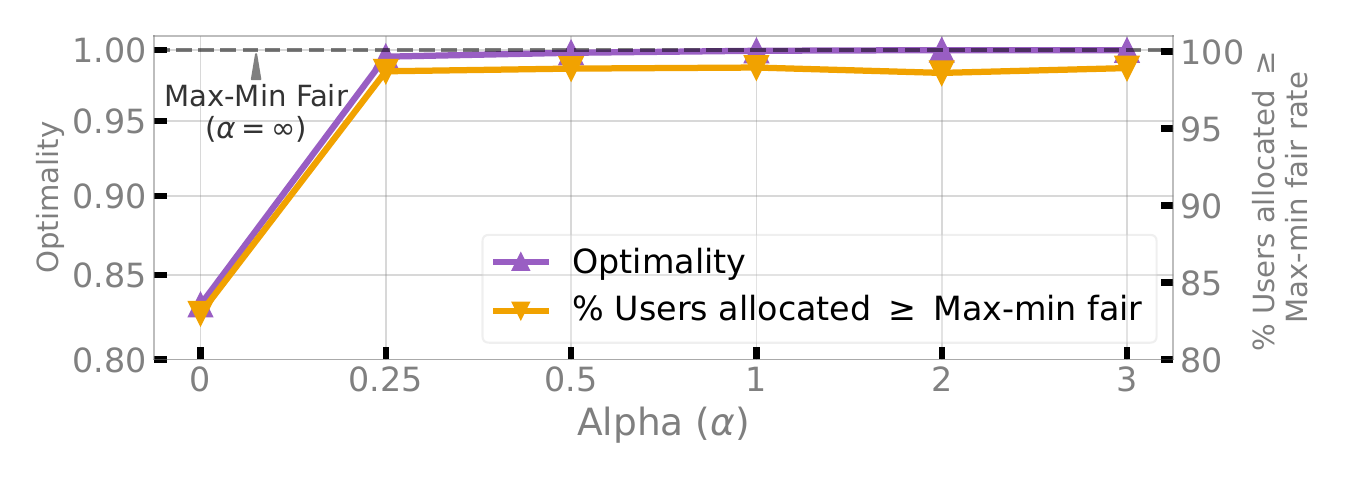}
    \caption{Optimality (relative to max-min fair) of $\alpha$-utility functions with increasing $\alpha$.}
    \label{fig:alpha-v-opt}
\end{figure}

\begin{figure}[ht]
    \centering
    \includegraphics[width=\linewidth]{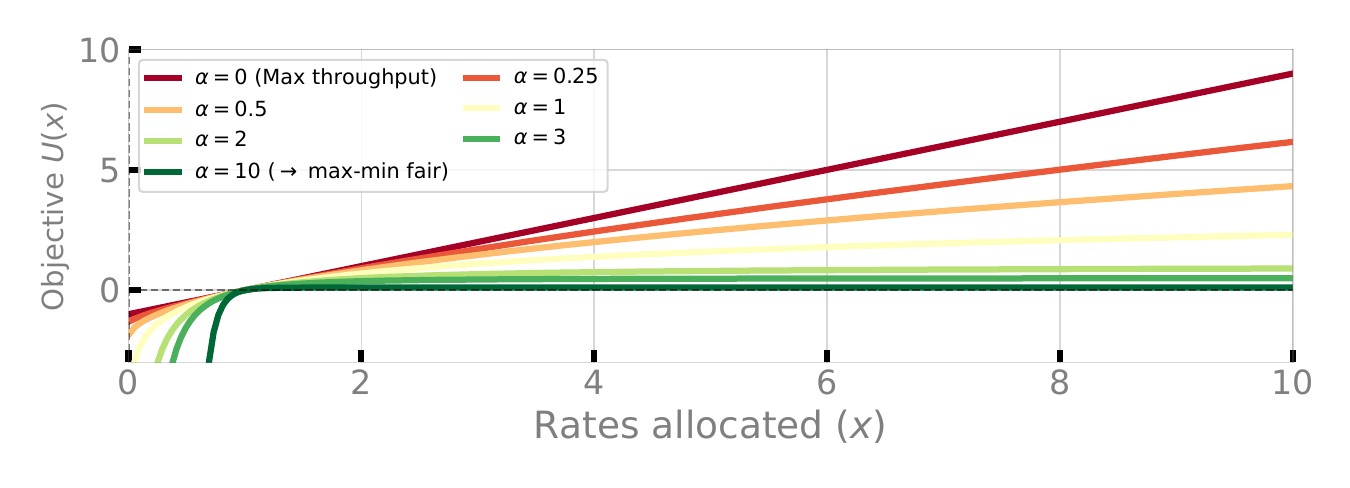}
    \caption{Curve of various $\alpha$-utility objective functions (larger $\alpha$ swings more).}
    \label{fig:x-v-alpha}
\end{figure}

\subsection{Sum Computation ($S_{st}$)}
\label{s:appendix:extended-design:sum-computation}
In \cref{s:design:iterative}, we noted that extracting the optimal path rates $x_r$ requires solving a vector maximization problem over the path rate vector $\mathbf{x} \in \mathbb{R}^{|P_{st}|}$. Unlike standard link constraints which operate element-wise, the $\alpha$-fair utility objective couples the rates of all paths supporting a given demand through their aggregate sum $S_{st} = \sum_r^{P_{st}} x$. 

For instance, under proportional fairness ($\alpha=1$), the structural form of the local subproblem is akin to $\log(\sum \mathbf{x}) - \sum \mathbf{x} + \sum (\mathbf{x}-\mathbf{y})^2$, differentiating which gives $\frac{1}{S_{st}} - S_{st} + 2x_r = 0$. While this equation cannot be solved isolated, we can take the system of equations for all $r \in P_{st}$ to get an equation of the form $|P_{st}| * \frac{1}{S_{st}} - |P_{st} * |S_{st} + 2S_{st} = 0$. In other words, we now have a quadratic equation in $S_{st}$. 




\paragraph{\textbf{Generic $\alpha$-Fairness (Newton-Bisection)}}
For a general $\alpha > 0$ objective, substituting the individual rate updates back into the definition of $S_{st}$ yields a polynomial equation of the form,
\begin{equation}
    A S^{\alpha-1} + B S + C = 0.
\end{equation}\label{eq:sum-computation} Because this function is strictly monotonic and convex over the physically valid domain of flow rates, it possesses a unique positive root. 

We compute this root using a parallelized Newton-Bisection method\cite{stoer2002numerical}. \sysname uses Newton-Raphson descent for rapid convergence but intelligently falls back to binary bisection if the gradient step overshoots the valid bounds. This guarantees exact convergence to $S_{st}^*$ within bounded number of GPU cycles.

\paragraph{\textbf{Proportional Fairness ($\alpha=1$) Quadric Root}}
Applying the standard quadratic formula yields the exact commodity sum instantly.

\subsection{Feasibility Projection}
\label{s:appendix:extended-design:projection}
\cref{alg:projection} describes the feasibility projection for \sysname.

In this section, we discuss potential modifications or extensions to the projection scheme.
\rbnew{A naive approach, such as uniformly scaling down all flows passing through a bottleneck, can disproportionately penalize flows that are already starved and undoes the optimization progress made by the subproblems. Instead, \sysname employs an $\alpha$-aware projection scheme, described in \cref{alg:projection}, which accounts for the operator's fairness objective.
}

\paragraph{\textbf{Parallel Execution and Over-Trimming}}
\rbnew{
To maintain \sysname's extremely low runtimes, the projection must be computed in parallel across thousands of GPU threads. However, synchronous parallel trimming introduces the risk of ``over-trimming'': if two independent congested links both decide to trim the same path simultaneously, the path loses twice the intended flow. While \cref{alg:projection} shows a simplified single-iteration view for clarity, our full implementation mitigates over-trimming by applying the calculated trims in small fractional steps over $T$ iterations.
}
\rbnew{
\paragraph{\textbf{Scoring for Max-Min Fairness}}
For max-min fairness ($\alpha \rightarrow \infty$), we disregard $C_p$ and use the value of the demand served $S_c$ as the score. Therefore, demands served higher get penalized more, in line with the max-min fair objective.
}

\begin{algorithm}[t]
\caption{\sysname's \textit{iterative, optimal} algorithm with $\alpha$-continuation.}
\label{alg:gate}
\small
\begin{algorithmic}[1]
\Statex\textbf{Input:} Link capacities ${C_e}$, demands ${D_{st}}$, pre-computed paths $\{P_{st}\}$, target fairness $\alpha_{target}$ (set to $\infty$ for max-min fair), and convergence threshold $\gamma$.
\Statex\textbf{Optional input:} Initial rates $x^0$, and learning rate $\beta$.
\Statex\textbf{Output:} Final allocation rates $x$.
\vspace{1.5mm}
\State \textbf{Initialization:}
\Statex \hspace{1.25em} $x^0 = \frac{D_{st}}{|P_{st}|}$, unless input is provided.
\Statex \hspace{1.25em} $y^0 = x^0$
\Statex \hspace{1.25em} $slack_{*} \leftarrow 0^+$ and $dual_{*} \leftarrow 0$.
\Statex \hspace{1.25em} $\alpha \gets 0$
\Statex \hspace{1.25em} $k \gets 0$
\vspace{2mm}
\While{$\alpha \le \alpha_{target}$} 
        \vspace{1mm}
        \State \textbf{Auxiliary updates:}
        \Statex \hspace{3em} $dual_*^{k+1} \gets \cref{eq:dual-updates-dem,eq:dual-updates-cap,eq:dual-updates-con,eq:dual-updates-neg}$ \Comment{dual updates}
        \vspace{1mm}
        \Statex \hspace{3em} $y_{e, r}^{k+1} \gets \cref{eq:rate-suggestion}$ \Comment{rate suggestion updates}
        \vspace{1mm}

        
        \State \textbf{Rate updates:}
        \Statex \hspace{3em} $S^{*}_{st} \gets \cref{eq:sum-computation}$ \Comment{Sum computation}
        \Statex \hspace{3em} $x_r^{k+1} \gets \cref{eq:rate-computation}$ 

        \vspace{1.5mm}
        \State \textbf{Convergence Optimizations:}
        \Statex \hspace{3em} $\beta^{k+1} \gets \cref{eq:adaptive-learning}$ \Comment{Learning rate updates}
        \If{residuals $(\mathcal{L}^{k + 1} - \mathcal{L}^k) \le \gamma$} 
            \If{$k_\alpha == k_{\alpha - 1}$}
                \State \textbf{break} \Comment{converged for $\alpha_{target}$}
            \Else 
            \State $\alpha \gets \alpha + 1$
            \Comment{converged for current $\alpha$}
            \EndIf
        \EndIf
        \State $k \gets k + 1$
    
\EndWhile
\vspace{1.5mm}
\State \textbf{Feasibility Projection}
\Statex \hspace{1.25em} $x \gets \text{Projection}(x^{k+1}, \alpha)$ \Comment{See Algorithm \ref{alg:projection}}
\State \Return $x$.
\end{algorithmic}
\end{algorithm}

\begin{algorithm}[t]
\caption{\sysname's $\alpha$-aware, parallelized, Projection.}
\label{alg:projection}
\small
\begin{algorithmic}[1]
\Statex\textbf{Input:} Path rates $x$ with possible constraint violations, link capacities $\{C_e\}$, demands $\{D_c\}$, fairness metric $\alpha$.
\Statex\textbf{Output:} Path rates $x$, with no constraint violations.

\State \textbf{Resolve negative rates:}
\State $x \gets \max(x, 0)$
\vspace{1mm}
\State \textbf{Auxiliary variables:} 
\Statex \hspace{1.25em} $S_c \gets \sum_{r \in P_c} x_r$, \quad $L_e \gets \sum_{r: e \in r} x_r$
\Statex \hspace{1.25em} Demand Violations: $E_c \gets \max(S_c - D_c, 0)$
\Statex \hspace{1.25em} Capacity Violations: $V_e \gets \max(L_e - C_e, 0)$

\Statex \hspace{1.25em} $C_p \gets \sum_{e \in p} \mathbb{I}(V_e > 0)$ \Comment{\# violations per-path}
\Statex \hspace{1.25em} $score_p \gets (S_{c(p)})^\alpha \times C_p$ \Comment{alpha-aware and violations-aware}

\vspace{1mm}
\State \textbf{Resolve demand violations:}
\For{each commodity $c$} \Comment{in parallel}
    \State Let $P_c^* = \{p \in P_c\}$, sorted in descending order of $score_p$.
    \State Iterate through $P_c^*$ until $E_c = 0$,
    \Statex \hspace{2.5em} $\delta x_p \gets min(E_c, x_p)$
    \Statex \hspace{2.5em} $x_{p} \gets x_p - \delta x_p$
    \Statex \hspace{2.5em} $E_c \gets E_c - \delta x_p$
\EndFor

\vspace{1mm}
\State \textbf{Recompute auxiliary variables:} ($S_c, L_e, V_e$, $C_p$, $score_p$)

\vspace{1mm}
\State \textbf{Resolve capacity violations:}
\For{each edge $e$} \Comment{in parallel}
    \State Let $P_e^* = \{p \in P \mid e \in p\}$, sorted in descending order of $score_p$.
    \State Iterate through $P_e^*$ until $V_e = 0$,
    \Statex \hspace{2.5em} $\delta x_{p} \gets \min(x_p, V_e)$
    \Statex \hspace{2.5em} $x_{p} \gets x_{p} - \delta x_p$
    \Statex \hspace{2.5em} $V_e \gets V_e - \delta x_{p}$
\EndFor

\State \Return $x$
\end{algorithmic}
\end{algorithm}

\subsection{Optimality Guarantees} \label{s:appendix:extended-design:optimality}
\begin{proof}
\osrdel{The proof is structured in two parts.} First, we map \sysname's formulation to a 2-block ADMM-style decomposition. Applying the standard ADMM convergence theorem ensures that for any fixed $\alpha$, the inner iterations strictly satisfy capacity and demand constraints while converging to the optimal $\alpha$-fair utility. Second, as $\alpha$ increments, the allocations track the homotopy path toward max-min fairness. We prove using Karush-Kuhn-Tucker (KKT) conditions that if the allocation stagnates across an $\alpha$-increment (triggering the zero-iteration halt condition), the network is physically locked by structural bottlenecks, mathematically guaranteeing the allocation has reached its max-min fair limit.

\noindent\textbf{Part I: ADMM Convergence for Fixed $\alpha$}

We map the Augmented Lagrangian in \cref{eq:augmented-lagrangian} to the standard 2-block ADMM template $\min_{u,v} F(u)+G(v)$ subject to $\mathcal A u+\mathcal B v=c$ and apply the classical convergence theorem for 2-block ADMM~\cite{neal2011distributed}. It suffices to verify: (a) closed, proper, and convex objectives, (b) existence of a primal--dual saddle point (e.g., Slater's condition), and (c) exact block minimizations.

\noindent\textbf{(a) Closed, proper, convex objective:} The indicator function for the non-negative orthant is closed, proper, and convex. For $\alpha>0$, the utility function $\varphi_\alpha$ is closed, proper, and convex on $(0,\infty)$ (extended by $+\infty$ for $S\le 0$), hence the aggregate objective $f_\alpha(x)=\sum_{(s,t)\in\mathcal K}\varphi_\alpha(S_{st}(x))$ is closed, proper, and convex on its domain. Therefore, the decomposed objective blocks $F$ and $G$ are closed, proper, and convex.

\noindent\textbf{(b) Slater point (strict feasibility):} Consider the original TE problem in minimization form:

$$ \begin{aligned} \min_{x}\ \sum_{s,t\in \mathcal K}-U_\alpha(\sum_{r\in P_{st}}x_r)\quad\text{s.t.}\quad x\ge 0,\\ S_{st}(x)\le D_{st}\ \forall(s,t)\in\mathcal K,\ \sum_{r:\,e\in r}x_r\le C_e\ \forall e. \end{aligned} $$

By our construction, $\mathcal K$ contains only commodities with $P_{st}\neq\emptyset$ and $D_{st}>0$, and every link $e$ that belongs to at least one path in $\bigcup_{(s,t)\in\mathcal K}P_{st}$ has $C_e>0$. For each $(s,t)\in\mathcal K$, let $m_{st}:=|P_{st}|$ and pick $\varepsilon>0$. Set:

$$ x_r^\circ:=\frac{\varepsilon}{m_{st}},\qquad \forall r\in P_{st},\ \forall (s,t)\in\mathcal K. $$

Then $S_{st}(x^\circ)=\varepsilon\in(0,D_{st})$ provided $\varepsilon<\min_{(s,t)\in\mathcal K}D_{st}$, so $x^\circ$ lies in the $\alpha$-domain interior ($S_{st}>0$). For each link $e$, the induced load is $\ell_e(\varepsilon):=\sum_{r:\,e\in r}x_r^\circ=\varepsilon\cdot \kappa_e$, where $\kappa_e:=\sum_{(s,t)\in\mathcal K}\frac{|\{r\in P_{st}: e\in r\}|}{m_{st}}<\infty$.

Choosing $\varepsilon < \min_{e:\kappa_e>0} C_e/\kappa_e$ yields strict feasibility for all inequalities and domain interiority. Hence Slater's condition holds, implying strong duality and existence of a saddle point.

\noindent\textbf{(c) Exact block minimizations:} The rate updates ($x_*$) are exact because non-negativity is rigorously enforced via the dual equation \cref{eq:dual-updates-neg}, and each commodity subproblem cleanly reduces to solving a unique scalar root in \cref{eq:sum-computation}. The rate suggestion updates ($y_*$) are exact because each link subproblem explicitly resolves into the closed-form bounded projection given in \cref{eq:rate-suggestion}. Therefore, all hypotheses of the standard 2-block ADMM convergence theorem strictly apply.
\smallskip

\noindent\textbf{Part II: Homotopy Continuation to Max-Min Fairness under $\gamma$-Tolerance.}

We now prove that if the continuation loop stagnates across an $\alpha$-increment (meaning the optimal allocation requires zero ADMM iterations to converge for the new $\alpha$), the allocation is within a strictly bounded neighborhood of the max-min fair limit. Fairness in resource allocation here is formally measured by equal raw allocation capped by demand.

Let the feasible routing space $\mathcal P$ be a convex polytope defined by the linear network capacity and demand constraints. Because the algorithm terminates when the primal and dual residuals fall below the convergence threshold $\gamma$, the exact Karush-Kuhn-Tucker (KKT) stationarity condition is relaxed to a $\gamma$-approximate condition. For a given $\alpha$, the approximate allocation $\tilde{S}_\alpha$ satisfies:

\begin{equation}
\label{eq:approx_kkt_alpha}
\Big\|(\tilde{S}_\alpha)^{-\alpha} - \sum_{i \in \text{active}}\lambda_{i,\alpha}A_i\Big\| \le \mathcal{O}(\gamma)
\end{equation}

where $A_i$ represents the active constraint normals and $\lambda_{i,\alpha} \ge 0$ represents the corresponding dual multipliers.

Assume the algorithm halts due to stagnation, meaning the allocation $\tilde{S}_\alpha$ requires zero iterations to satisfy the $\gamma$-threshold for the incremented objective $\alpha+1$. Applying the relaxed KKT condition to the new objective yields:

\begin{equation}
\label{eq:approx_kkt_alpha_plus_1}
\Big\|(\tilde{S}_\alpha)^{-(\alpha+1)} - \sum_{i \in \text{active}}\lambda_{i,\alpha+1}A_i\Big\| \le \mathcal{O}(\gamma)
\end{equation}

For both $\gamma$-bounded conditions to hold simultaneously on the identical allocation vector $\tilde{S}_\alpha$ while the gradient undergoes a non-linear transformation, the active constraint normals $A_i$ must strictly dominate the gradient direction within the error bounds. Geometrically, rather than being rigidly pinned at a single vertex, the allocation $\tilde{S}_\alpha$ is bounded within an $\epsilon$-neighborhood of a fully saturated face or vertex of the convex polytope $\mathcal P$, where $\epsilon$ is directly proportional to $\gamma$.

In this $\gamma$-locked state, no feasible direction exists to increase the rate of any commodity without violating the $\gamma$-tolerance of the physical constraints. Every commodity is either fully satisfied and bounded by its absolute demand constraint ($S_{st} \approx D_{st}$), or it is constrained by a saturated bottleneck link where it holds an equal raw share of the residual capacity (up to the $\gamma$-margin). Thus, exact stagnation under a $\gamma$-threshold implies the algorithm has discovered the physical saturation point of the topology within a rigorously bounded error.
\end{proof}

\section{Extended Evaluation}
\label{sec:appendix:evaluation}
\subsection{Drift Adjusted Optimality.}
\label{s:appendix:evaluation:dao-definition}
\osr{We note in \cref{s:evaluation:metrics} that DAO applies the penalty incurred due to solver runtime to the solver optimality. \cref{fig:dao-trend} unpacks DAO into runtime and optimality, by reporting all three metrics for SWAN and Danna over \wanlarge. The optimality of the TE solver, marked by the hollow circles, is the quality of the allocation computed by the TE solver. This is a property of the TE solver alone, and by construction, does not vary with runtime. The \emph{dashed curve} is the \emph{feasible optimality} of the same allocation as a function of when it is outputted:
for a hypothetical runtime $K$ (every 10ms), we compute $DAO_K$, i.e. the optimality that survives once
$X$ is trimmed to the network state $(D,C)_K$. Each curve therefore begins at its solver's optimality as
$K \to 0$ and decays from there, isolating the penalty due to staleness from the penalty due to potential approximate nature of TE solvers (for example, SWAN). Finally, the \emph{dotted vertical line} marks the solver's measured runtime on \wanlarge{}, and the \emph{solid circle}, where that runtime intersects the feasible optimality curve, is the solver's DAO.}

\osr{We make two observations about the degradation of feasible optimality in \cref{fig:dao-trend}. First, the degradation is solver-specific, i.e. the same runtime costs different \textit{amount} of optimality to different TE solvers. What is lost depends not only on how long the solver ran, but on the allocation it produced (for example, packed more tightly vs spread out). A deeper study of the nature of allocation to the penalty incurred would be valuable, left for future work. Second, the degradation is non-linear in runtime. Instead, it appears to be closer to a higher-order polynomial in the low runtime range ($<10^2$), and starts to plateau thereafter.} 

\osr{\paragraph{Why do we compute DAO against $OPT_0$?} We note in \cref{s:evaluation:metrics} that DAO is computed against $OPT_0$, the optimum for the state the solver was handed, rather than $OPT_K$, the optimum for the state at the time it returns. We do this for two reasons: First, it ensures a consistent baseline across TE solvers. Comparing against $OPT_K$, with different $K$ for each TE solver,  would not allow us to compare different DAOs. At the same time, it is not practical to attempt to equalize end times, since it requires knowing each solver's runtime in
advance. Second, $OPT_0$ is the best any solver could have done with the same input, so measuring against it allows us to capture the entire suboptimality of the solver, i.e. what it gave up to (i) algorithm approximation, and (ii) solver runtime.}

\begin{figure*}[t]
\centering
\captionsetup{font=footnotesize, labelfont=normalfont, textfont=normalfont} 
  \begin{minipage}[b]{0.30\textwidth}
    \includegraphics[width=\linewidth]{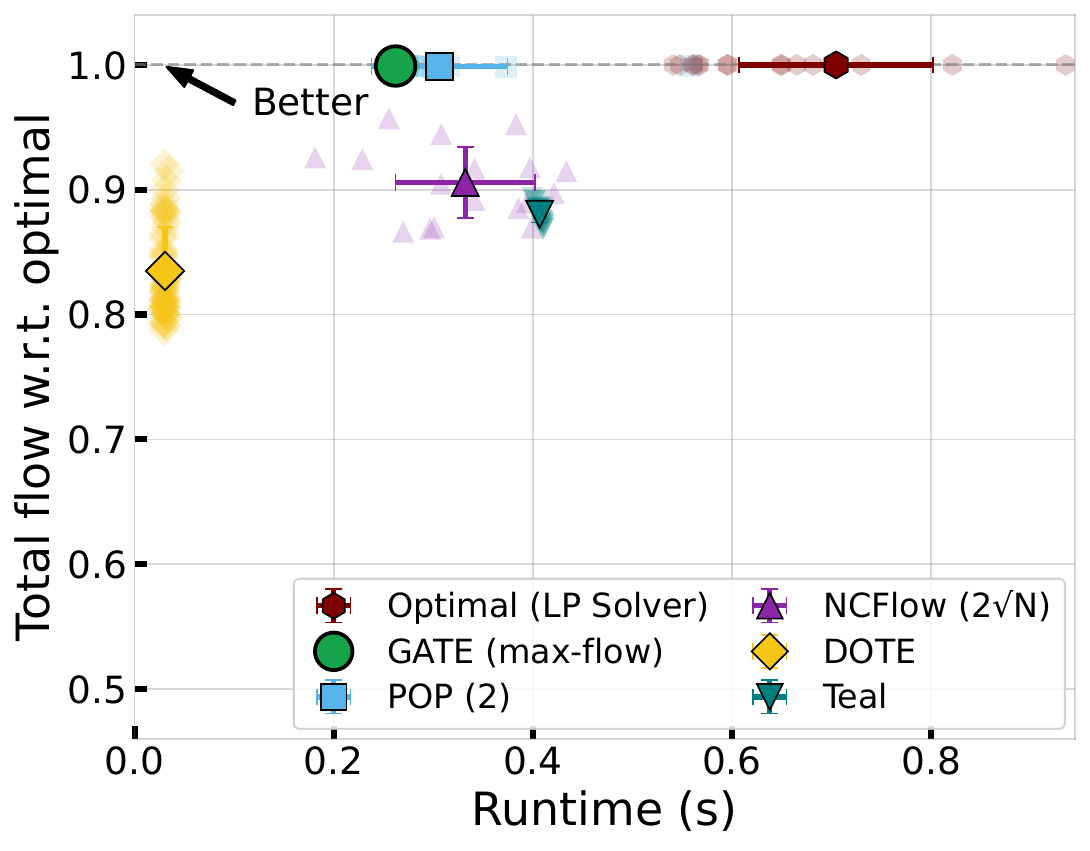}
    \caption{\osr{\wanlarge with 1x traffic demand (large network, low load).}}
    \label{fig:runtime-v-throughput-b2-1x}
  \end{minipage}
  \hfill
  \begin{minipage}[b]{0.30\textwidth}
    \includegraphics[width=\linewidth]{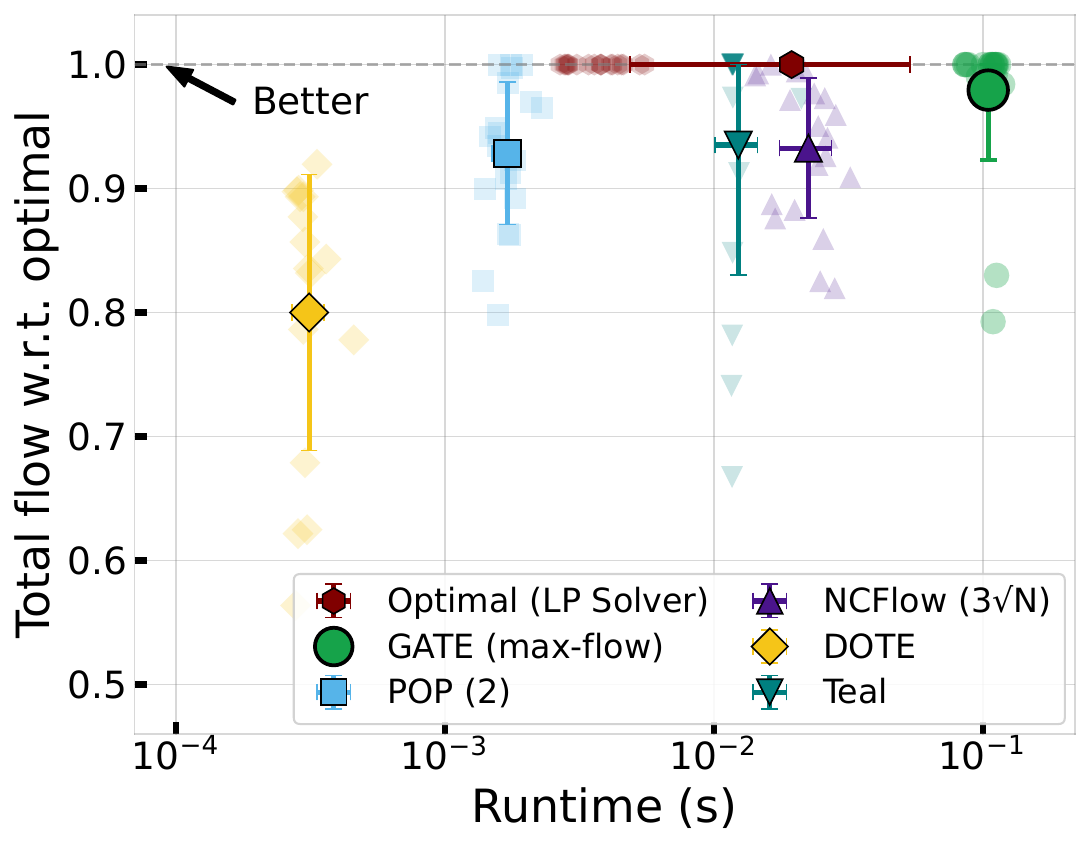}
\caption{\osr{Abilene with 1x traffic demand (small network, low load).}}
\label{fig:runtime-v-throughput-abilene-1x}
  \end{minipage}
  \hfill
    \begin{minipage}[b]{0.30\textwidth}
    \includegraphics[width=\linewidth]{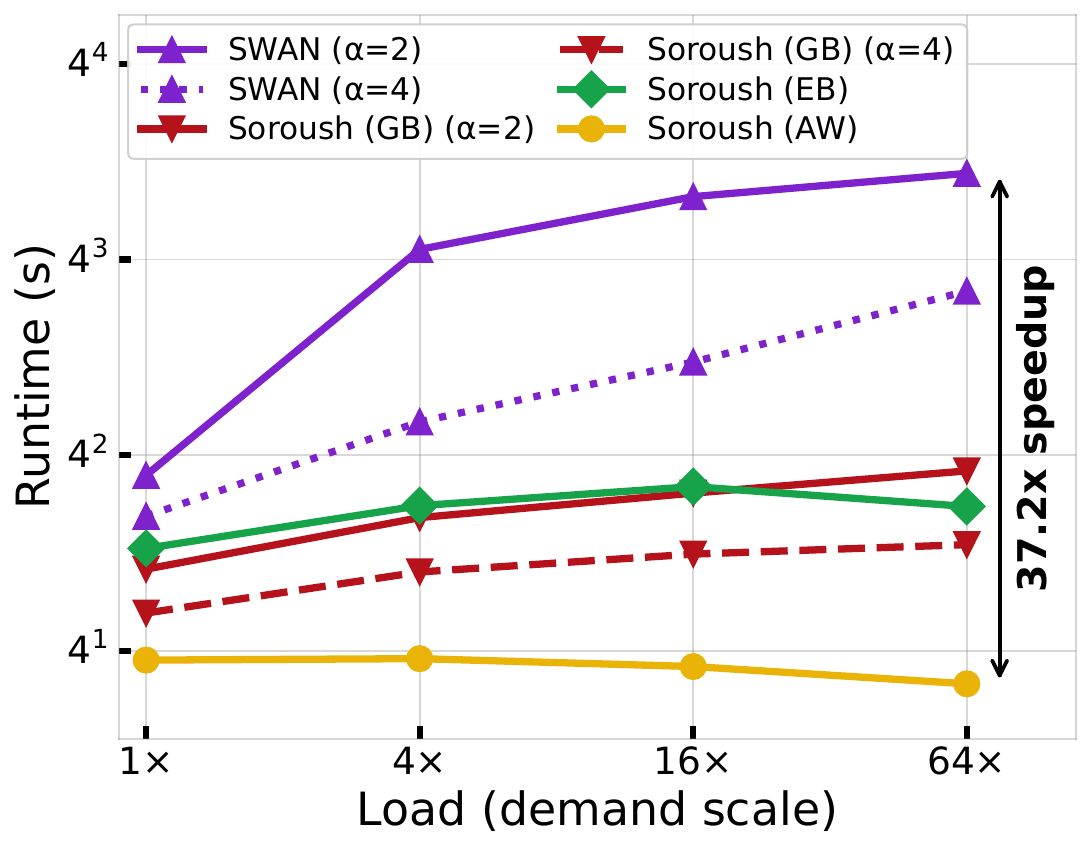}
    \caption{\osr{Runtimes of different parameters and variants of Soroush and SWAN with
    increasing traffic load.}}
    \label{fig:runtime-v-load-tuning}
  \end{minipage}
  \caption{Runtime and Optimality for max-flow TE solvers with different network sizes and traffic load configurations.}
\end{figure*}

\begin{figure*}[t]
\centering
\captionsetup{font=footnotesize, labelfont=normalfont, textfont=normalfont} 
  \begin{minipage}[b]{0.31\textwidth}
    \includegraphics[width=\linewidth]{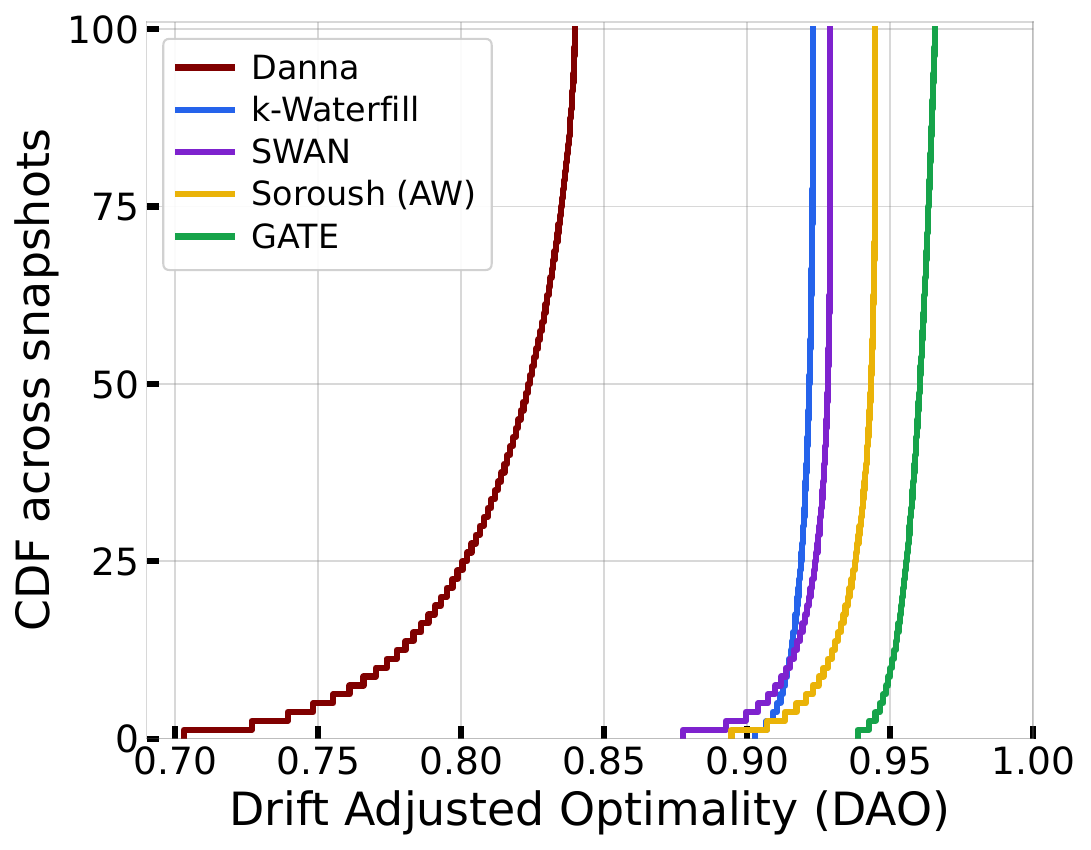}
    \caption{CDF of Drift Adjusted Optimality \osr{(DAO) over \wanlarge.}\osrdel{ experienced over 80 snapshots of \wanlarge.}}
    \label{fig:dao-cdf}
  \end{minipage}
  \hfill
  \begin{minipage}[b]{0.31\textwidth}
    \includegraphics[width=\linewidth]{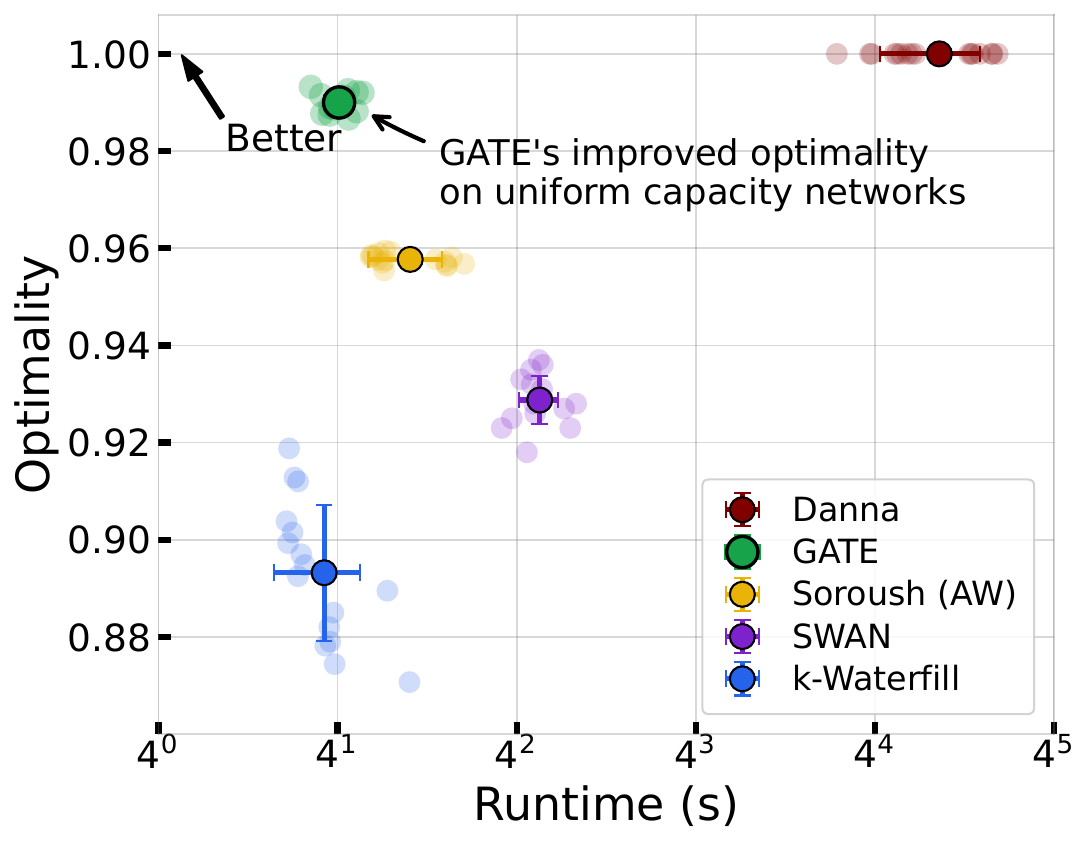}
    \caption{Runtime \osrdel{and}\osr{\&} Optimality for TE \osrdel{algorithms}\osr{solvers} over Cogentco with uniform link capacities.}
    \label{fig:runtime-v-optimality-cogentco}
  \end{minipage}
  \hfill
    \begin{minipage}[b]{0.31\textwidth}
    \includegraphics[width=\linewidth]{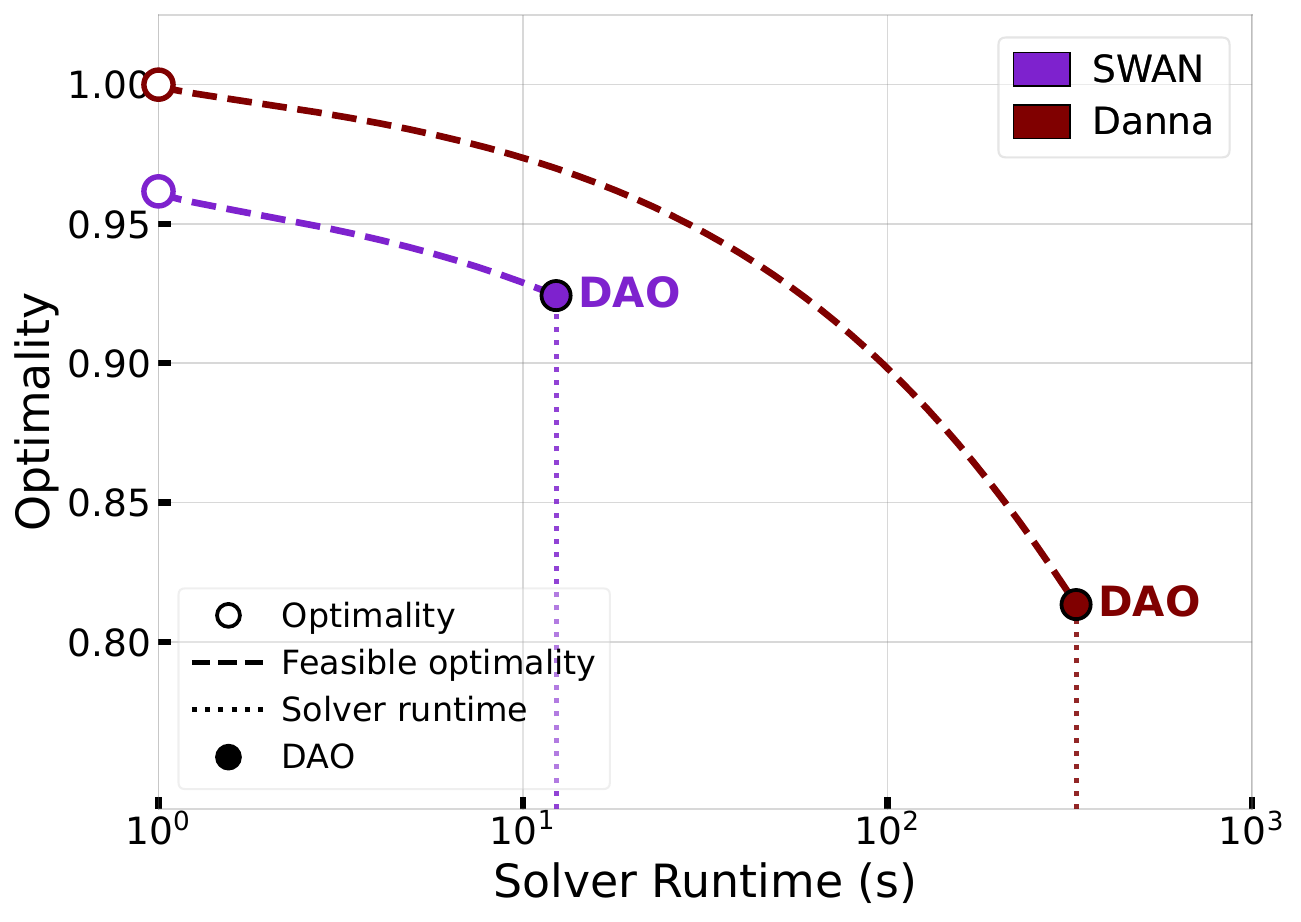}
    \caption{\osr{Runtime, Optimality, and DAO for SWAN and Danna over \wanlarge.}}
    \label{fig:dao-trend}
  \end{minipage}
\end{figure*}

\begin{figure}[ht]
    \centering
    \includegraphics[width=\linewidth]{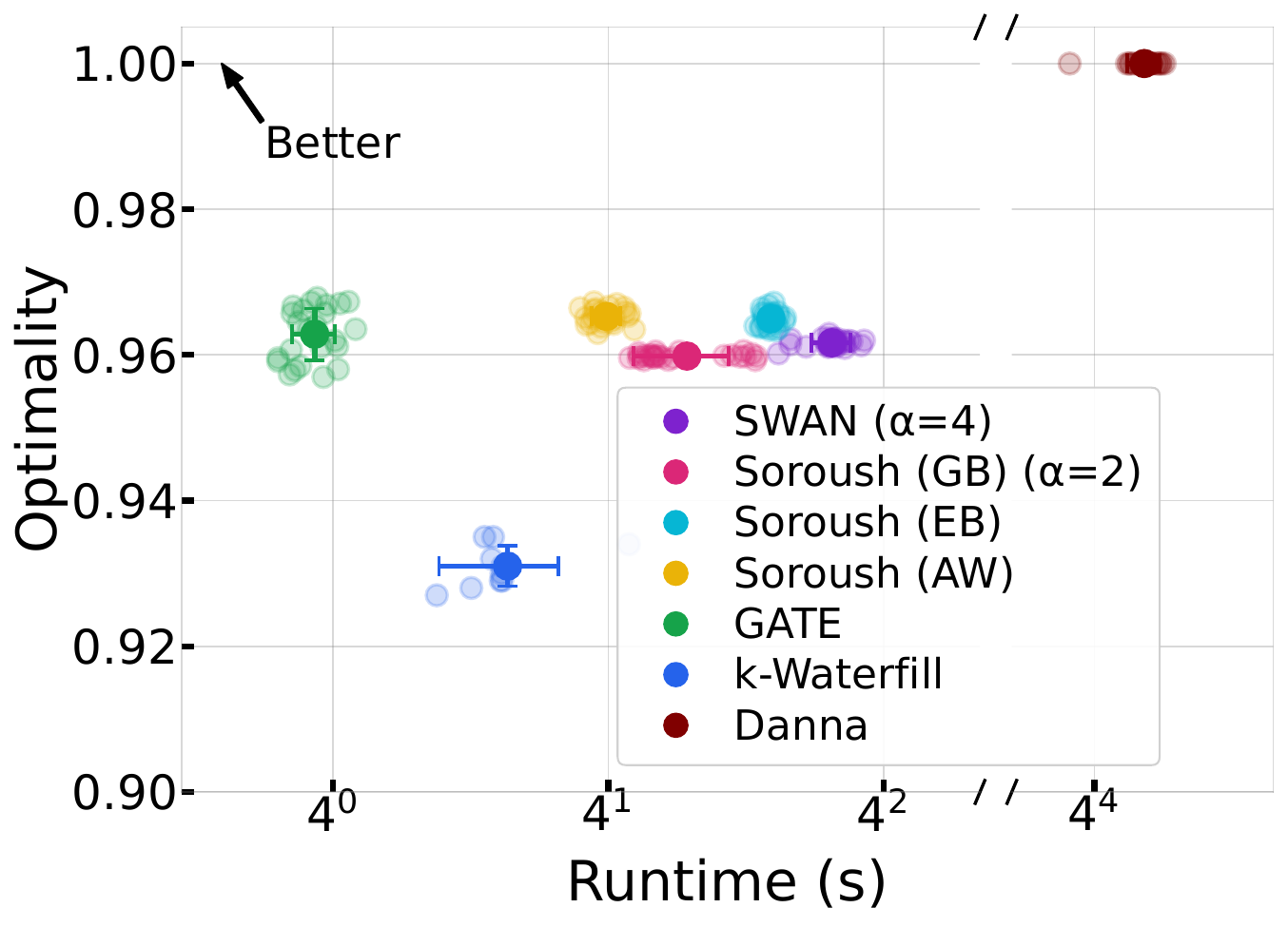}
    \caption{Runtime and Optimality of max-min fair TE solvers (including all variants of Soroush) over \wanlarge at production load. Soroush AW achieves a better DAO than the other two variants.}
    \label{fig:runtime-v-optimality-all-soroush}
\end{figure}

\subsection{Benchmark Tuning}
\label{s:appendix:evaluation:tuning}
\srnew{should the title be "Solver Tuning"? Or "Tuning Solver Performance"}
\osr{\noindent\textbf{Tuning procedure.} Many of the TE solvers in our benchmark, expose parameters that trade-off between optimality and runtime. For each solver, we sweep the parameter values listed in \cref{t:benchmark-tuning}, and use the configuration that achieves the highest DAO. We perform this for every topology and load, since the best operating point shifts with network size and congestion.}

\osr{We tune for DAO, since it captures both runtime and optimality, as elaborated in \cref{s:appendix:evaluation:dao-definition}. Note that this can place a baseline at a different operating point than the one reported in its own paper; we discuss one such case in \cref{s:appendix:evaluation:result-consistency}.}

\osr{Note that Danna~\cite{Danna} computes the exact max-min fair allocation and exposes no approximation parameter. DOTE~\cite{perry2023dote} and Teal~\cite{teal} are trained offline; we use the configurations reported in the papers. For Soroush, in addition to tuning a variant, we further select the best variant of the three -- Geometric Binner, Adaptive Waterfill, and Equi-Depth Binner, which achieves the higest DAO. For completion, we present a version of \cref{fig:runtime-v-optimality-b2} with all Soroush variants in \cref{fig:runtime-v-optimality-all-soroush}. We found that Soroush's Adaptive Waterfill typically gets better DAO compared to the other two. For \sysname, we tune $\gamma$ and $A$ for \sysname, with a fixed $C$, as discussed in \cref{s:design:implementation}.}

\begin{table}[t]
\centering
\small
\renewcommand{\arraystretch}{1.15}
\rowcolors{1}{gray!30}{gray!10}
\begin{tabular}{@{}lll@{}}
\textbf{TE Solver} & \textbf{Parameter tuned} & \textbf{Range swept} \\ \hline
\sysname            & $\gamma$ & $10^{-2}, 10^{-4}, 10^{-6}$ \\
 & $A$ & ${2, 5, 10}$ \\
SWAN~\cite{swan}    & $\alpha$ (rate granularity)    & $\{1.5, 2, 4, 8\}$ \\
Soroush (GB)~\cite{soroush} & $\alpha$ (bin separation) & $\{1.5, 2, 4, 8\}$ \\
POP~\cite{pop}      & $k$ (\# subproblems)           & $\{2, 4, 8, 16\}$ \\
NCFlow~\cite{ncflow} & \# clusters                   & $\{\sqrt{N}, 2\sqrt{N}, 3\sqrt{N}\}$ \\
DeDe~\cite{DeDe}    & $\rho$ (penalty)               & $10^{x}\hfill \forall x \in [-4, 3]$ \\
PDMCF~\cite{zhang2025solving} & $\gamma$ (conv. thresh.) & $10^{-3}, 10^{-4}$ \\
Lin~\cite{lin2006utility} & $K$ (\# iterations)                & $10^x \hfill \forall x \in [2, 5]$ \\
\end{tabular}
\caption{\osr{Range of each parameters swept for each TE solver.}}
\label{t:benchmark-tuning}
\end{table}

\subsection{Differences in evaluation setup relative to prior work} 
\label{s:appendix:evaluation:result-consistency}

\osr{The performance of TE solvers reported in \cref{fig:runtime-v-optimality-b2} might differ compared to results reported in prior works~\cite{soroush, swan}. We highlight two differences in our evaluation setup,}

\osr{\noindent\textbf{TE solvers tuned for DAO.} As reported in \cref{s:evaluation:setup}, we tune TE solvers to the configuration that achieves the highest DAO. In contrast, prior works\cite{soroush} tuned benchmarks for lowest runtimes, or used constant parameters ($\alpha=2$ for SWAN). Interestingly, DAO-aware tuning selects a larger $\alpha$ for SWAN, since the runtime saving ($2\times$) outweighed the low drop in optimality. Soroush (GB), however, did not benefit from the higher $\alpha$ configurations because of its already low runtimes (recall the runtime penalty to DAO is poly-time-like at low runtimes (\cref{s:appendix:evaluation:dao-definition})). We describe parameter tuning in more detail in \cref{s:appendix:evaluation:tuning}.}

\vspace{0.2em}
\textbf{Unscaled production traffic demands.} 
\osr{\cref{fig:runtime-v-optimality-b2} uses the
demands observed in production, which are not severely congested. We observe that different TE solvers scale differently as network congestion increases. \cref{fig:runtime-v-load-B2} reports that at $64\times$ the production demand, Soroush (AW) is $37.2\times$ faster than the fixed $\alpha=2$, consistent with results from prior-work\cite{soroush}.}

\srnew{Is there / do we need / some takeaway or additional statement to the effect that, if we remove these differences, the results will match?}

\subsection{\osr{Max-flow across network sizes and loads}}
\label{s:appendix:evaluation:max-flow}

\osr{\cref{s:evaluation:runtime-v-optimality} reports max-flow results for \wanlarge{}
under $16\times$ scaled demand, the configuration in which solving max-flow is relatively harder
. We also report configurations with production-load traffic and smaller networks. We observe the trend that the massive-parallelism benefits of \sysname become even more beneficial on larger networks, at higher loads.}

\vspace{0.2em}
\noindent\osr{\textbf{\sysname has \emph{slightly} better performance with (unscaled) production traffic, on a large network.} \ak{confusing heading -- slightly better than what?}
\cref{fig:runtime-v-throughput-b2-1x} reports the throughput and runtime of max-flow TE solvers over \wanlarge, identical to \cref{fig:runtime-v-throughput-b2-16x}, except with unscaled traffic, observed in production. While \sysname still attains the highest DAO, serving the optimal
throughput in $0.26$s, but the margin is quite small: POP is within $0.05$s and also reaches
near-optimal throughput. While the LP solver is $2.5\times$ slower than \sysname, the absolute runtime different is fairly small. The runtime of learning-based algorithms (DOTE, and Teal) is independent of the traffic load, and remains nearly unchanged. As the network becomes congested, for example, in \cref{fig:runtime-v-throughput-b2-16x} \srnew{broken}, the runtime of TE solvers that leverage LP solvers increases, and the relative performance of \sysname reflects better.}

\vspace{0.2em}
\noindent \osr{\textbf{On small networks, \sysname has much longer \emph{relative} runtimes.}
\cref{fig:runtime-v-throughput-abilene-1x} reports Abilene ($11$ nodes) at production load. \sysname takes $0.1$s, slower than every other max-flow TE solver. This is the same effect we report for max-min fairness in
\cref{s:evaluation:scaling-behavior}, where \sysname is slower than SWAN and Soroush on
sub-$30$-router topologies, since the problem is too small to amortize the overheads of dispatching jobs across GPU threads and data synchronization. The effect is more severe for max-flow compared to max-min-fair since the linear
objective makes every LP solver based baselines much faster.}

\srnew{Looking across the topic sentences in bold, I'm not sure what our takeaway is: that GATE is not really useful for max-flow?}
\vspace{0.2em}
\noindent\osr{\textbf{A note on topologies.} The max-flow solvers we compare
against~\cite{pop, ncflow, teal, perry2023dote} report their largest speedups on the
biggest synthetic topologies available in public datasets, which are larger than the
production WANs we evaluate. Partition-based approaches in particular gain most when a
network can be divided into many large, weakly-coupled pieces. Our results should
therefore be read as characterizing these solvers on a deployed WAN, rather than as
contradicting the speedups reported in their respective papers.}


\subsection{\osr{Continued anaylsis of DAO of TE solvers}}
\label{s:appendix:evaluation:dao-cdf}
\osr{\sysname's shorter runtime enables it to perform disproportionately better in the worst scenarios, which is visible in the lower percentiles of the CDF of DAO in \cref{fig:dao-cdf}, corresponding to cases of large capacity changes. This is because we both reconverge quickly, and to a better state. Note that each data point in this CDF is a single snapshot; this actually dilutes some of \sysname's fast reaction benefits, because a link failure only affects some flows in a snapshot. Fast reaction to failure is more directly captured by runtime}.

\subsection{\osr{Performance of additional decomposition-based TE solvers}}
\label{s:appendix:evaluation:decomposition-te}
\osr{We compare PDMCF\mbox{\cite{zhang2025solving}} and Lin\mbox{\cite{lin2006utility}} for completion, despite their limitations described in \mbox{\cref{sec:motivation-existing}}.  Lin\mbox{\cite{lin2006utility}} converges to a severely suboptimal ($\geq 20\%$) allocation, with $6x$ longer runtimes than \sysname. Its performance substantially degrades on heterogenous networks (\wansmall and \wanlarge). PDMCF\mbox{\cite{zhang2025solving}} serves more throughput than \sysname on all but \wansmall and \wanlarge, since it expects infinite demand and infinite paths in its formulation. However, it takes nearly $10x$ longer to converge on the same threshold $\gamma=10^{-3}$. On heterogenous networks (\wansmall and \wanlarge), PDMCF stagnates at a suboptimal value much lower than \sysname, despite the relaxed constraints. Note that we equip DeDe, PDMCF, and Lin with \sysname's projection scheme to fairly compare the decomposition technique.}
\vspace{0.2em}
\subsection{Looking inside GATE's Iterations}
\label{s:appendix:evaluation:gate-iters}
We continue the analysis of GATE's internal progress from \cref{s:evaluation:iterates-behavior}. In the first $\approx 20$ iterations, \sysname makes steady progress filling in the network with traffic. For nearly the same period, the learning rate $\beta$ also increases, since \bgnew{path} rates are increasing (primal residuals) faster than constraint violations are increasing (dual residuals).

\bgnew{Around 20 iterations, network utilization has filled up enough that increasing rates starts to violate constraints (\cref{fig:gate-internals} (b) and (c)). Beyond 20 iterations, the Lagrangian variables are such that rate changes tend to balance between} resolving constraint violations and maximizing the objective. Interestingly, while the percentage of constraints violated does not decrease significantly, the mean relative violation drops sharply and becomes less than 1\% within 75 iterations.

After \textasciitilde50 iterations, the primal and dual residuals fall below $\gamma$, i.e. $10^{-3}$. At this point, \sysname increases $\alpha$, in order to progress towards \bgnew{max-min fairness}. Note that $\alpha = 0$ itself achieves nearly 90\% of optimal. \sysname takes fewer iterations to converge on $\alpha=1$, as much of the progress has already been made.
After $\approx 85$ iterations, $\alpha$ increases to $2$, and progresses for $<10$ iterations before increasing again. At this point, \bgnew{on the first iteration with $\alpha=3$, \sysname remains converged (residuals $<\gamma$). This is \sysname's usual stopping criterion. In the figure, we show what would happen if we continue anyway.  Note that while optimality} increases by a small amount even after the stopping point, the additional runtime required results in an effective loss in DAO. 